\documentclass[11pt,a4paper,reqno]{amsart}%
\usepackage{amsthm,amsmath,amsfonts,amssymb,amsxtra,appendix,bookmark}

\theoremstyle{plain}
\newtheorem{theorem}{Theorem}[section]
\newtheorem{lemma}[theorem]{Lemma}
\newtheorem{corollary}[theorem]{Corollary}
\newtheorem{proposition}[theorem]{Proposition}

\theoremstyle{definition}

\theoremstyle{remark}
\newtheorem{remark}{Remark}

\numberwithin{equation}{section}

\DeclareMathOperator{\ess}{ess}

\DeclareMathOperator{\Tr}{Tr}

\def\geqslant{\ge}
\def\leqslant{\le}
\def\bq{\begin{eqnarray}}
\def\eq{\end{eqnarray}}
\def\bqq{\begin{eqnarray*}}
\def\eqq{\end{eqnarray*}}
\def\nn{\nonumber}
\def\minus {\backslash}
\def\eps{\varepsilon}
\def\wto{\rightharpoonup}
\newcommand{\norm}[1]{\left\lVert #1 \right\rVert}

\renewcommand{\epsilon}{\varepsilon}

\def\cF {\mathcal{F}}

\def\cH{\mathcal{H}}

\def\R {\mathbb{R}}
\def\C {\mathbb{C}}
\def\N {\mathcal{N}}
\def\U {\mathbb{U}}

\def\E {\mathcal{E}}
\def\cE {\mathcal{E}}

\def\cK{\mathcal{K}}

\def\F {\mathcal{F}}

\def\H{\gH}

\def\V {\mathcal{V}}

\def\R {\mathbb{R}}
\def\C {\mathbb{C}}
\def\N {\mathcal{N}}
\def\U {\mathbb{U}}

\def\S {\mathcal{S}}

\def\G {\mathcal{G}}
\def\E {\mathcal{E}}
\def\A{\mathcal{A}}

\def\K{\mathcal{K}}

\def\d{{\rm d}}
\newcommand{\gH}{\mathfrak{H}}
\newcommand\ii{{\ensuremath {\infty}}}
\newcommand\pscal[1]{{\ensuremath{\left\langle #1 \right\rangle}}}
\newcommand{\bH}{\mathbb{H}}
\newcommand{\dGamma}{{\ensuremath{\rm d}\Gamma}}
\title{Bogoliubov spectrum of interacting Bose gases}

\author[M. Lewin]{Mathieu Lewin}
\address{CNRS \& Universit\'e de Cergy-Pontoise, Mathematics Department (UMR 8088), F-95000 Cergy-Pontoise, France} 
\email{mathieu.lewin@math.cnrs.fr}

\author[P.~T. Nam]{Phan Th\`anh Nam}
\address{CNRS \& Universit\'e de Cergy-Pontoise, Mathematics Department (UMR 8088), F-95000 Cergy-Pontoise, France} 
\email{phan-thanh.nam@u-cergy.fr}

\author[S. Serfaty]{Sylvia Serfaty}
\address{Universit\'e Pierre et Marie Curie, Laboratoire Jacques-Louis Lions (UMR 7598), F-75005 Paris, France; and Courant Institute, New York University, 251 Mercer Street, New York NY 10012, USA} 
\email{serfaty@ann.jussieu.fr}

\author[J.~P. Solovej]{Jan Philip Solovej}
\address{Department of Mathematical Sciences, University of Copenhagen, Universitetsparken 5, DK-2100 Copenhagen \O, Denmark} 
\email{solovej@math.ku.dk}

\begin{document}
\date{{March 10, 2014.}\ Final version to appear in \textit{Comm. Pure Appl. Math.}}
\maketitle

\begin{abstract} We study the large-$N$ limit of a system of $N$ bosons interacting with a potential of intensity $1/N$. When the ground state energy is to the first order given by Hartree's theory, we study the next order, predicted by Bogoliubov's theory. We show the convergence of the lower eigenvalues and eigenfunctions towards that of the Bogoliubov Hamiltonian (up to a convenient unitary transform). We also prove the convergence of the free energy when the system is sufficiently trapped. Our results are valid in an abstract setting, our main assumptions being that the Hartree ground state is unique and non-degenerate, and that there is complete Bose-Einstein condensation on this state. Using our method we then treat two applications: atoms with ``bosonic'' electrons on one hand, and trapped 2D and 3D Coulomb gases on the other hand. 
\end{abstract}

\setcounter{tocdepth}{2}
\tableofcontents

\section{Introduction}

In a famous paper \cite{Bogoliubov-47b}, Bogoliubov was able to predict the excitation spectrum of a quantum gas satisfying the Bose statistics and he used this to understand its superfluid behavior. Since Bogoliubov's work, there has been several attempts to formulate Bogoliubov's theory in a mathematically rigorous way. This was especially successful for completely integrable 1D systems~\cite{Girardeau-60,LieLin-63,Lieb-63,CalMar-69,Calogero-71,Sutherland-71,Sutherland-71b}, for the ground state energy of one and two-component Bose gases~\cite{LieSol-01,LieSol-04,Solovej-06}, and for the Lee-Huang-Yang formula of dilute gases~\cite{ErdSchYau-08,GiuSei-09,YauYin-09}. Recently, Seiringer~\cite{Seiringer-11} and Grech-Seiringer~\cite{GreSei-12} have for the first time justified Bogoliubov's theory for the excitation spectrum of trapped Bose gases, with a general short range interaction, in the mean-field regime. See, e.g.,~\cite{Yngvason-10} for a recent review on the subject and~\cite{CorDerZin-09} for a discussion of translation-invariant systems. 

The purpose of this article is to give general conditions under which Bogoliubov's theory is valid, that is, predicts the lowest part of the spectrum of the many-body Hamiltonian of bosons, in the mean-field regime. Our results cover a very large class of interacting Bose gases and they generalize the recent works~\cite{Seiringer-11,GreSei-12}. In particular, our method applies to Coulomb systems.

We consider a system of $N$ quantum particles, described by the Hamiltonian 
\[
H_N= \sum\limits_{i = 1}^N T_{x_i} + \frac{1}{N-1} \sum\limits_{1 \leqslant i < j \leqslant N} {w(x_i-x_j)},
\]
acting on the symmetric (a.k.a.~bosonic) space 
$$\H^N=\bigotimes_{\text{sym}}^N L^2(\Omega)$$
of square-integrable functions $\Psi\in L^2(\Omega^N)$ which are symmetric with respect to exchanges of their variables, namely
$$\Psi(x_{\sigma(1)},...,x_{\sigma(N)})=\Psi(x_1,...,x_N).$$
for every $\sigma$ in the permutation group $\mathfrak{S}_N$. Here $\Omega$ is an open subset of $\R^d$ with $d\ge 1$, $T$ is a self-adjoint operator on $L^2(\Omega)$ with domain $D(T)$, and $w$ is an even real-valued function describing the interactions between particles. We have neglected spin for convenience, but it can be added without changing any of our result. The Hamiltonian $H_N$ describes a system of $N$ bosons living in $\Omega$. 

The operator $T$ can contain both the kinetic energy and an external potential which is applied to the system, including possibly a magnetic field. We typically think of $T=-\Delta$ on a bounded set $\Omega$ with appropriate boundary conditions (Dirichlet, Neumann or periodic), or of $T=-\Delta+V(x)$ on $\Omega=\R^d$, with $V$ an external potential which serves to 
bind the particles. In the latter case, the function $V$ could tend to zero at infinity but it then has to be sufficiently negative somewhere, or it could tend to infinity at infinity, in which case all the particles are confined. We could also replace the non-relativistic operator $-\Delta$ by its relativistic counterpart $\sqrt{1-\Delta}-1$. We shall keep the operator $T$ sufficiently general in this paper, such that all these situations are covered. The function $w$ could also be replaced by an abstract two-body operator but we do not consider this here for simplicity.

We are interested in the limit of a large number $N$ of particles. Here we are considering the \emph{mean-field regime}, in which the interaction has a fixed range (the function $w$ is fixed) but its intensity is assumed to tend to zero in the limit $N\to\ii$, hence the factor $1/(N-1)$ in front of the interaction term in the Hamiltonian $H_N$. This factor makes the two sums of order $N$ in $H_N$ and, in this case, an important insight is given by \emph{Hartree theory}. 

Let us recall that a Hartree state is an uncorrelated many-body wave function in which all of the particles live in the same state $u\in L^2(\Omega)$ such that $\int_\Omega|u|^2=1$, and which takes the form
$$\Psi(x_1,...,x_N)=u(x_1)\cdots u(x_N).$$
The energy of such a state is
$$\langle \Psi, H_N\Psi\rangle=N\left(\langle u, Tu\rangle+\frac12D(|u|^2,|u|^2)\right):=N\,\E_{\rm H}(u)$$
where
$$D(f,g):=\int_{\Omega}\int_{\Omega}\overline{f(x)}\,g(y)w(x-y)\,\d x\,\d y$$
is the classical interaction. Henceforth, all the Hilbert spaces we consider have inner products which are conjugate linear in the first variable and linear in the second. 

Provided that there is \emph{Bose-Einstein condensation}, the leading term of the ground state energy 
$$E(N):=\inf {\rm spec}\, H_N$$ 
is given by Hartree's theory:
$$\boxed{E(N)=Ne_{\rm H}+o(N),}$$ 
where $e_{\rm H}$ is the corresponding Hartree ground state energy:
\bq \label{eq:Hartree-energy}
e_{\rm H} :=\inf_{\substack{u \in {L^2}({\Omega})\\||u|| = 1}} \,\E_{\rm H}(u)= \inf_{\substack{u \in {L^2}({\Omega})\\||u|| = 1}} \left\{ {\left\langle {u,Tu} \right\rangle  + \frac{1}{2}D(|u{|^2},|u{|^2})} \right\}.
\eq  
In this paper, we shall assume that there exists a unique Hartree minimizer $u_0$ for $e_{\rm H}$. It is then a solution of the nonlinear Hartree equation
\bq \label{eq:Hartree-equation}
0=\big(T+|u_0|^2\ast w-\mu_{\rm H}\big)u_0:=h\,u_0,
\eq
where $\mu_{\rm H}\in \mathbb{R}$ is a Lagrange multiplier.

Bogoliubov's theory predicts the next order term (of order $O(1)$) in the expansion of the ground state energy $E(N)$. It also predicts the leading term and the second term for the lower eigenvalues of $H_N$. The Bogoliubov method consists in describing variations of the wavefunctions around the Hartree state $u_0\otimes\cdots\otimes u_0$ in a suitable manner. We will explain this in detail in Section~\ref{sec:Bogoliubov_Hamiltonian} below. The final result is an effective Hamiltonian $\mathbb{H}$, called the \emph{Bogoliubov Hamiltonian} which is such that the lower spectrum of $H_N$ in $\H^N$ is given, in the limit $N\to\infty$, by the spectrum of the effective operator $Ne_{\rm H} + \mathbb{H}$, up to an error of order $o(1)$. Formally, we therefore find that
\begin{equation}
\boxed{H_N\simeq N\,e_{\rm H}+\mathbb{H}+o(1).}
\label{eq:expansion_Bogoliubov_intro}
\end{equation}
This vague statement is made precise in our main result, Theorem~\ref{thm:lower-spectrum} below. 

The essential fact about the Bogoliubov Hamiltonian $\mathbb{H}$ is that it is a \emph{non particle conserving} self-adjoint operator  acting on the Fock space
$$\cF_+:=\C\oplus \bigoplus_{n\geq1}\bigotimes_{\text{sym}}^n\gH_+$$
where
$$\gH_+=\{u_0\}^\perp\subset L^2(\Omega)$$
is the one-body space of excited particles. We started with a particle-conserving model and we end up with a theory in Fock space, in which the number of particles is not fixed. The reason is that we are describing here the \emph{excitations} around the reference Hartree state. 

For the acquainted reader, we mention that $\mathbb{H}$ is indeed nothing but the second-quantization of (half) the Hessian of the Hartree energy at $u_0$, and its expression in a second quantized form is
\begin{multline}
\mathbb{H}:=\int_{\Omega}a^*(x)\big(h\,a\big)(x)\,\d x+\frac12\int_{\Omega}\int_{\Omega} w(x-y)\bigg(2u_0(x)\overline{u_0(y)}a^*(x)a(y)\\+u_0(x)u_0(y)a^*(x)a^*(y)+ \overline{u_0(x)}\overline{u_0(y)}a(x)a(y)\bigg)\d x\,\d y.
\label{eq:Bogoliubov_Hamiltonian} 
\end{multline}
Here $a^*(x)$ is the creation operator of an excited particle at $x$, acting in the Fock space $\cF_+$, and $h$ is defined in (\ref{eq:Hartree-equation}). We will explain the meaning of this formula later in Section~\ref{sec:Bogoliubov_Hamiltonian}.

A result similar to~\eqref{eq:expansion_Bogoliubov_intro} has recently been obtained for weakly interacting Bose gases by Seiringer \cite{Seiringer-11} and Grech-Seiringer \cite{GreSei-12}. They assumed that $w$ is bounded, decays fast enough and has non-negative Fourier transform. The operator $T$ was $T=-\Delta$ in a box with periodic boundary conditions in~\cite{Seiringer-11} and $T=-\Delta+V(x)$ on $\R^d$ with $V(x)\to+\ii$ at infinity in~\cite{GreSei-12}. Our method is different from that of~\cite{Seiringer-11, GreSei-12} and it applies to a larger class of models.

We give in this paper a list of abstract conditions that a Bose gas should satisfy in order to get the Bogoliubov result~\eqref{eq:expansion_Bogoliubov_intro}. These conditions are given and explained in Section~\ref{sec:conditions} below. Loosely speaking, we assume that there is \emph{complete Bose-Einstein condensation} on a unique Hartree minimizer $u_0$ which we assume to be non-degenerate. Our message is that, once the Bose-Einstein condensation is proved, one can get the next order in the expansion of the energy by Bogoliubov's theory. No further assumption is needed.

The paper is organized as follows. In the next section, we define our model by giving the appropriate assumptions on $T$ and $w$, and we properly define the Bogoliubov Hamiltonian $\mathbb{H}$. We then state our main results, Theorem~\ref{thm:lower-spectrum} and Theorem \ref{thm:positive-temperature}. In Sections~\ref{sec:bosonic-atoms} and~\ref{sec:2D-log-gas}, we apply our abstract result to two particular examples: bosonic atoms and trapped Coulomb gases. Sections~\ref{sec:bounds-truncated-Fock-space}--\ref{sec:proof-Main-Theorems} are devoted to the proof of the main abstract results. 

\bigskip

\noindent\textbf{Acknowledgement.} We thank Thomas S\o rensen and an anonymous referee for helpful comments and corrections. M.L. and P.T.N. acknowledge financial support from the European Research Council under the European Community's Seventh Framework Programme (FP7/2007-2013 Grant Agreement MNIQS 258023). S.S. was supported by a EURYI award. J.P.S. was supported by a grant from the Danish Council for Independent Research $|$ Natural Sciences.

\section{Main abstract results}

In this section we state our main result. 

\subsection{Assumptions}\label{sec:conditions}

We start by giving the main assumptions on $T$ and $w$, under which our results apply. Later in Section~\ref{sec:applications} we consider two specific examples, for which all the following assumptions are satisfied.

The first condition concerns the properties of $T$ and $w$ which are necessary to give a proper meaning to the many-body Hamiltonian $H_N$.

\medskip

\noindent (A1) (One- and two-body operators). {\it The operator $T:D(T)\to L^2(\Omega)$ is a densely defined, bounded from below, self-adjoint operator. The function $w:\R^d \to \mathbb{R}$ is Borel-measurable and $w(x)=w(-x)$. Moreover, there exist constants $C>0$, $1>\alpha_1>0$ and $\alpha_2>0$ such that}
\bq \label{eq:bound-w-1}
-\alpha_1 ( T_x + T_y +C) \le w(x-y) \le \alpha_2 ( T_x + T_y +C ) \qquad  {\rm on}~L^2(\Omega^2).
\eq

\medskip

Note that, although we keep the one-body operator $T$ abstract, we use a two-body operator $w$ which is a translation-invariant multiplication operator in $L^2(\Omega^2)$. This is only for convenience. All our results are also valid if $w$ is an abstract two-body operator on $\gH^2$ which satisfies an estimate of the same type as~\eqref{eq:bound-w-1}, and if $\gH=L^2(\Omega)$ is an abstract separable Hilbert space. However, in this case the expressions of the Hartree energy and of the corresponding nonlinear equations are different (they cannot be expressed using a convolution). We shall not consider this abstract setting to avoid any confusion.

Under Assumption (A1), $H_N$ is bounded from below,
\bq \label{eq:HN>=T}
H_N\geq (1-\alpha_1)\sum_{i=1}^N T_i-CN.
\eq
In the paper we always work with the Friedrichs extension~\cite{ReeSim2}, still denoted by $H_N$. Note that we do not assume the positivity or boundedness of $w$ or its Fourier transform, but only that it is relatively form-bounded with respect to $T$.

\medskip

Our second assumption is about Hartree theory.

\medskip

\noindent (A2) (Hartree theory). {\it  The variational problem (\ref{eq:Hartree-energy}) has a unique (up to a phase) minimizer $u_0$ in the quadratic form domain $Q(T)$ of $T$. Moreover, $u_0$ is non-degenerate in the sense that
\bq \label{eq:hK>=eta}
\left( {\begin{array}{*{20}{c}}
  {h + K_1}&K_2 \\ 
  K_2^*&{\overline{h} + \overline{K_1}} 
\end{array}} \right) \geqslant \eta_{\rm H}\qquad \text{on}\quad \gH_+\oplus \overline{\gH}_+
\eq
for some constant $\eta_{\rm H}>0$, where $\gH_+=\{u_0\}^\perp$. Here 
$$h:=T+ |u_0|^2*w-\mu_{\rm H},$$
with the Lagrange multiplier $\mu_{\rm H}:=e_{\rm H}+D(|u_0|^2,|u_0|^2)/2$ (which ensures that $hu_0=0$), and $K_1: \gH_+\to \gH_+$ and $K_2=\overline{\gH}_+ \to \gH_+$ are operators defined by
\bqq
\langle u, K_1 v\rangle &=& \int_{\Omega}\int_{\Omega} \overline{u(x)} v(y) u_0(x) \overline{u_0(y)} w(x-y)  \d x \d y,\hfill\\
\langle u, K_2 \overline{v}\rangle &=& \int_{\Omega}\int_{\Omega} \overline{u(x)} \overline{v(y)}u_0(x) u_0(y) w(x-y)  \d x \d y
\eqq
for all $u,v\in \gH_+$. The operators $K_1$ and $K_2$ are assumed to be Hilbert-Schmidt, that is}
\bq
\int_{\Omega}\int_{\Omega}|u_0(x)|^2|u_0(y)|^2w(x-y)^2\,\d x\,\d y<\infty.
\label{eq:cond_K}
\eq

\begin{remark} \it Note that once we have assumed that $u_0$ is a minimizer for (\ref{eq:Hartree-energy}), then we always have the mean-field equation $hu_0=0$ due to a standard argument (see e.g. \cite[Theorem 11.5]{LieLos-01}). In the whole paper we will for simplicity use the same notation for the operator $h$ on the full one-body space $\gH$ and for its restriction to the smaller space $\gH_+$. For an operator $A$ on $\gH_+$, the notation $\overline{A}$ means $JA J$ with $J$ being the complex conjugate, namely $\overline{A}(v)=\overline{A (\overline{v})}$ for every $v\in \gH_+$.
\end{remark}

\begin{remark} \it While we shall treat $K_1$ as a one-body operator, we should really think of $K_2$ as its integral kernel $K_2(x,y)= \big(Q\otimes Q\big) (u_0\otimes u_0 w(.-.))(x,y)$, which is the two-body function obtained by projecting the symmetric function $u_0(x)u_0(y)w(x-y)$ onto $\gH_+^2$.
\end{remark}

In (A2) we are making assumptions about the uniqueness and non-degeneracy of the Hartree ground state $u_0$. The Hessian of the Hartree energy can easily be seen to be
\begin{align}
&\frac12{\rm Hess}\;\cE_{\rm H}(u_0)(v,v)\nonumber\\
&~=\pscal{v,hv}+\frac12\int_{\Omega}\int_{\Omega} w(x-y)\bigg(\overline{v(x)}u_0(x)\overline{u_0(y)}v(y)+v(x)\overline{u_0(x)}u_0(y)\overline{v(y)}\nonumber\\
&~\qquad+\overline{v(x)}u_0(x)u_0(y)\overline{v(y)}+ v(x)\overline{u_0(x)}\overline{u_0(y)}v(y)\bigg)\d x\,\d y\nonumber\\
&~=\frac12\pscal{\begin{pmatrix}v\\ \overline{v}\end{pmatrix},\begin{pmatrix}
h+K_1&K_2\\ K_2^*&\overline{h}+\overline{K_1}\end{pmatrix}\begin{pmatrix}v\\ \overline{v}\end{pmatrix}}_{\gH_+\oplus \overline{\gH}_+}
\label{eq:Hessian}
\end{align}
for all $v\in\gH_+$. It turns out that the non-degeneracy of the Hessian,
$${\rm Hess}\;\cE_{\rm H}(u_0)(v,v)\geq 2\eta_{\rm H}\norm{v}_{L^2(\Omega)}^2,$$
is equivalent to our assumption~\eqref{eq:hK>=eta}. When $u_0$ is real, as it is in many applications, then $\overline\gH_+=\gH_+$, $K_1=K_2$ and it can be verified (using a test function of the form $(v,-v)$) that~\eqref{eq:hK>=eta} implies $h\geq \eta_{\rm H}$ on $\gH_+$, which means that there is a gap above the first eigenvalue $0$. In general, we however only know that $h+K_1\geq\eta_{\rm H}$. 

The Hilbert-Schmidt assumption~\eqref{eq:cond_K} on $K_1$ and $K_2$ will be useful later to ensure that the Bogoliubov Hamiltonian $\bH$ is well defined (see Section~\ref{sec:Bogoliubov_Hamiltonian} below).

\medskip

Our last assumption is about the validity of Hartree theory in the limit $N\to\infty$. We assume that the system condensates in the unique Hartree ground state $u_0$. This assumption will be necessary for the proof of the lower bound on the spectrum of $H_N$.

\medskip

\noindent(A3) (Complete Bose-Einstein condensation). {\it For any constant $R>0$, there exists a function $\eps_R: \mathbb{N}\to [0,\infty)$ with $\lim_{N\to \infty} \eps_R(N)=0$ such that, for any wave function $\Psi_N\in \H^N$ satisfying $\left\langle {\Psi_N ,{H_N}\Psi_N } \right\rangle \le E(N)+R$, one has 
\begin{equation}
\frac{\left\langle { u_0,  \gamma_{\Psi_N} u_0} \right\rangle}{N} \ge 1-\eps_R(N)
\label{eq:condensation}
\end{equation}
where $u_0$ is the Hartree minimizer in Assumption {\rm (A2)}.}
\medskip

Here $\gamma_{\Psi}$ is the one-body density matrix of the wave function $\Psi\in\gH^N$, which is the trace-class operator on $L^2(\Omega)$ with kernel
\[\gamma_{\Psi} (x,y) := N\int\limits_{{\Omega^{N-1}}} {\Psi (x,{x_2},...,{x_N})\overline {\Psi (y,{x_2},...,{x_N})} \d{x_2}...\d{x_N}}. \]
Note that a Hartree state has the density matrix
$\gamma_{u^{\otimes N}}=Nu(x)\overline{u(y)}$.
Therefore~\eqref{eq:condensation} is the same as saying that 
$\gamma_{\Psi_N}$ is in some sense close to $\gamma_{u_0^{\otimes N}}$. For more explanation about  Bose-Einstein condensation, we refer to the discussion in~\cite{LieSeiSolYng-05}.

In many practical situations, the complete Bose-Einstein condensation (A3) follows from the uniqueness of the Hartree ground state in (A2). This is discussed in the recent work~\cite{LewNamRou-13}, based on a compactness argument which does not provide any explicit error estimate.

For Coulomb systems (see Section~\ref{sec:applications}), a stronger condensation property with an explicit error estimate will hold true. Namely, we will have a bound from below valid for all $\Psi_N\in \H^N$, and not only for those which have a low energy. We therefore introduce the following stronger assumption, which obviously implies (A3):

\medskip

\noindent (A3s) (Strong condensation). {\it We have $h\geq\eta_{\rm H}>0$ on $\gH_+$, and there exists a constant $0<\varepsilon_0<1$ such that
$$ H_N -N e_{\rm H} \ge (1-\varepsilon_0) \sum_{j=1}^Nh_{j} + o(N).$$ 
Here $h$ is the mean-field operator given in Assumption {\rm (A2)}.}
\medskip

In fact, in practice  (A3s) follows from a Lieb-Oxford inequality 
\[  \left\langle {\Psi ,\left( {\sum\limits_{1 \leqslant i < j \leqslant N} {{w_{ij}}} } \right)\Psi } \right\rangle  \geqslant \frac{1}{2}D({\rho _\Psi },{\rho _\Psi }) + {\rm error}\]
where $\rho_\Psi(x)=\gamma_\Psi(x,x)$. It is for proving estimates of this form that it is often useful to know that $\widehat{w}\geq0$ where $\widehat{\cdot}$ denotes the Fourier transform. We will come back to this in Section~\ref{sec:applications} where we consider two examples. 

\subsection{The Bogoliubov Hamiltonian}\label{sec:Bogoliubov_Hamiltonian}

Near the Hartree minimizer we have
\begin{equation}
\cE_{\rm H}\left(\frac{u_0+v}{(1+\norm{v}^2)^{1/2}}\right)=\cE_{\rm H}(u_0)+\frac12 {\rm Hess}\;\cE_{\rm H}(u_0)(v,v)+o\big( \langle v,(T+C)v\rangle \big)
\label{eq:expansion_Hartree} 
\end{equation}
for any $v$ which is orthogonal to $u_0$, that is, $v\in\gH_+$. The next order in the expansion of the eigenvalues of the Hamiltonian $H_N$ will be given by the Bogoliubov Hamiltonian $\bH$, which is obtained by second quantizing the Hessian in~\eqref{eq:expansion_Hartree}. More precisely, this means replacing $\overline{v(x)}$ by an operator $a^*(x)$ which creates an excited particle at $x$, and $v(x)$ by an operator $a(x)$ which annihilates it. These operators (formally) act on the Fock space of excited particles
$$\cF_+:=\C\oplus \bigoplus_{n\geq1}\bigotimes_{\text{sym}}^n\gH_+=\C\oplus \bigoplus_{n\geq1}\gH^n_+.$$
So the expression of the Bogoliubov Hamiltonian is
\begin{multline}
\mathbb{H}:=\int_{\Omega}a^*(x)\big( (h+K_1)\,a\big)(x)\,\d x\\+\frac12\int_{\Omega}\int_{\Omega} \bigg( {K_2(x,y)}a^*(x)a^*(y)+ \overline{K_2(x,y)}a(x)a(y)\bigg)\d x\,\d y.
\label{eq:Bogoliubov}
\end{multline}

In order to make the formula (\ref{eq:Bogoliubov}) more transparent, let us explain how the  Hamiltonian $\bH$ acts on functions of $\cF_+$. If we have a $\psi_k\in\gH_+^k$, with $k\geq2$, then we get
\begin{equation}
\bH\psi_k=\cdots0\oplus\underbrace{\psi'_{k-2}}_{\in\gH_+^{k-2}}\oplus\,0\oplus\underbrace{\left(\sum_{j=1}^k(h+K_1)_{j}\right)\psi_k}_{\in\gH_+^{k}}\oplus\,0\oplus\underbrace{\psi'_{k+2}}_{\in\gH_+^{k+2}}\oplus\,0\cdots
\label{eq:Bogoliubov-rigorous} 
\end{equation}
where 
$$\psi'_{k+2}(x_1,...,x_{k+2})=\!\frac{1}{\sqrt{k!(k+2)!}}\!\sum_{\sigma\in\mathfrak{S}_{k+2}}\!\!K_2(x_{\sigma(1)},x_{\sigma(2)})\,\psi_k(x_{\sigma(3)},...,x_{\sigma(k+2)}),$$
$$\psi'_{k-2}(x_1,...,x_{k-2})=\sqrt{k(k-1)}\int_\Omega \d x_{k-1}\int_\Omega \d x_{k}\, \overline{K_2(x_{k-1},x_{k})}\,\psi_k(x_{1},...,x_{k}).$$

The link between the formal expression~\eqref{eq:Bogoliubov} and the rigorous formula~\eqref{eq:Bogoliubov-rigorous} is explained in~\cite[Sec. 1]{Lewin-11}. See also (\ref{eq:quadratic-Hamiltonian}) in Appendix \ref{apd:Bogoliubov-Hamiltonian} for another equivalent expression of $\bH$ using one-body density matrices.

Let us remark that for $\psi'_{k+2}$ to be in $L^2(\Omega^{k+2})$ for all $\psi_k\in L^2(\Omega^k)$, it is necessary and sufficient to have $K_2(.,.)$ in $L^2(\Omega^2)$, which is the same as assuming that $K_2$ is a Hilbert-Schmidt operator, as required in Assumption (A2).

Since $K_1$ and $K_2$ are Hilbert-Schmidt, the Hamiltonian $\bH$ is well defined on states living in truncated Fock spaces and in the domain of $h$:
\begin{equation}
\bigcup_{M\geq0}\bigoplus_{n=0}^M\bigotimes_{\text{sym}}^n D(h).
\label{eq:core-quadratic-Bogoliubov} 
\end{equation}
The following theorem tells us that $\bH$ is bounded from below and it is a well-defined self-adjoint operator by the Friedrichs method with the form domain being the same as that of $\dGamma(1+h)$ on $\cF_+$. Here we have used the usual notation $\dGamma(A)$ for the second quantization in Fock space of an operator $A$ acting on the one-body space $\gH_+$:
\bq \label{eq:def-dGamma}
\dGamma (A):=\bigoplus_{m=0}^\infty \sum_{j=1}^m A_j = \int_{\Omega}a^*(x)\big(A\,a\big)(x)\,\d x.
\eq
\begin{remark} \it We shall always denote by $C>0$ some (large) constant which depends only on $T$ and $w$. Two $C$'s in the same line may refer to different constants.
\end{remark}
\begin{theorem}[Bogoliubov Hamiltonian]\label{thm:Bogoliubov-Hamiltonian} If {\rm (A1)-(A2)} hold true, then the operator $\mathbb{H}$ is symmetric and on the core domain~\eqref{eq:core-quadratic-Bogoliubov} one has
\bq \label{eq:comparison-bH-dGamma(h+1)}
C^{-1}\dGamma (h+1) -C \le \bH \le  \dGamma (h+C) +C .
\eq
Consequently, the form domain of its Friedrichs extension (still denoted by $\bH$) is the same as that of $\dGamma(1+h)$ on $\cF_+$:
\begin{equation}
\bigoplus_{n=0}^\ii\bigotimes_{\text{sym}}^n D((h+1)^{1/2}).
\label{eq:form-domain-quadratic-Bogoliubov} 
\end{equation}
Moreover, we have the following spectral properties.

\smallskip

\noindent $(i)$ \emph{(Ground state and ground state energy).} The Hamiltonian $\bH$ admits a lowest eigenvalue which is simple. It is 
strictly negative, except if $w=0$, in which case we have $\inf \sigma (\mathbb{H})=0$ (the ground state being the vacuum in $\cF_+$). 

\smallskip

\noindent $(ii)$ \emph{(Essential spectrum).} The essential spectra of $h$ and $\bH$ are simultaneously empty or non empty, and we have in the latter case
\begin{equation*}
\sigma_{\rm ess} (\mathbb{H})= \sigma(\mathbb{H}) + \sigma_{\rm ess}(h).
\label{eq:ess-spectrum} 
\end{equation*}
Consequently, $\inf \sigma_{\rm ess}(\bH)-\inf \sigma(\bH)= \inf \sigma_{\rm ess}(h)\ge \eta_{\rm H}>0$.
\smallskip

\noindent $(iii)$ \emph{(Lower spectrum).} Assume that $\overline{T}=T$ (in this case $u_0$ is a real-valued function and hence $K_1=K_2$). If $h+K_1$ has infinitely many eigenvalues below its essential spectrum, then $\mathbb{H}$ also has infinitely many eigenvalues below its essential spectrum. 

On the other hand, if $\overline{T}=T$, $K_1\ge 0$ and $h$ has only finitely many eigenvalues below its essential spectrum, then $\mathbb{H}$ also has finitely many eigenvalues below its essential spectrum. 
\end{theorem}

We refer to Appendix \ref{apd:Bogoliubov-Hamiltonian} for a proof of Theorem~\ref{thm:Bogoliubov-Hamiltonian} and further discussions.

\begin{remark}\it Note that since $h+K_1\ge \eta_{\rm H}$ due to (\ref{eq:hK>=eta}) and $K_1$ is Hilbert-Schmidt, we always have $\inf \sigma_{\rm ess}(h)\ge \eta_{\rm H}>0$. 
\end{remark}

\begin{remark}\it
The reader should be cautious with the fact that, when $w\ne 0$, even though $K_2$ is a Hilbert-Schmidt operator on the one-body Hilbert space, the pairing term 
$$
\frac12\int_{\Omega}\int_{\Omega} \bigg(\overline{K_2(x,y)}a^*(x)a^*(y)+ K_2 (x,y)a(x)a(y)\bigg)\d x\,\d y
$$
is \emph{neither} bounded on $\F_+$, or relatively compact with respect to $\dGamma(h)$. Indeed, when the essential spectrum of $h$ is non empty, we have 
$$\inf \sigma_{\rm ess}(\bH)=\inf \sigma(\bH) + \inf \sigma_{\rm ess}(h)<\inf \sigma_{\rm ess}(h)=\inf \sigma_{\rm ess}(\dGamma (h))$$
due to Theorem \ref{thm:Bogoliubov-Hamiltonian}, and hence $\bH$ and $\dGamma(h)$ have different essential spectra.
\end{remark}

\subsection{Exciting the Hartree state}\label{sec:exciting-Hartree}

In this section we explain how to appropriately describe the variations of a many-body wave function $\Psi$ in the neighborhood of the Hartree state $u_0^{\otimes N}$. We will see that the Fock space $\cF_+$ arises naturally. We will also define a unitary operator $U_N$ which will be essential for our main statements.

For two symmetric functions $\Psi_k\in \gH^k$ and $\Psi_\ell\in\gH^\ell$, we recall that the symmetric tensor product is defined by
\begin{multline*}
\Psi_k\otimes_s\Psi_\ell(x_1,...,x_{k+\ell})\\=\frac{1}{\sqrt{k!\ell!(k+\ell)!}}\sum_{\sigma\in\mathfrak{S}_N}\Psi_k(x_{\sigma(1)},...,x_{\sigma(k)})\Psi_\ell(x_{\sigma(k+1)},...,x_{\sigma(k+\ell)}). 
\end{multline*}
Note that this tensor product satisfies the commutative property $\Psi_k\otimes_s\Psi_\ell=\Psi_\ell\otimes_s\Psi_k$ and the associative property $(\Psi_k\otimes_s\Psi_\ell)\otimes_s \Psi_m=\Psi_k\otimes_s(\Psi_\ell\otimes_s \Psi_m)$ for all $\Psi_k\in \gH^k$, $\Psi_\ell\in\gH^\ell$ and $\Psi_m\in \gH^m$. Consider now any (real-valued) orthonormal basis $u_0,u_1,...$ of $L^2(\Omega)$, containing the Hartree minimizer $u_0$. Then, it is known that $\{u_{i_1}\otimes_s\cdots \otimes_s u_{i_N}\}$ is an orthogonal basis of the symmetric space $\H^N$, where
$$u_{1}\otimes_s\cdots \otimes_s u_{N}(x_1,...,x_N):=\frac{1}{\sqrt{N!}}\sum_{\sigma\in\mathfrak{S}_N}u_{\sigma(1)}(x_1)\cdots u_{\sigma(N)}(x_N).$$
So we can write our many-body Hilbert space $\H^N$ as a direct sum of spaces 
$$\H^N=\mathfrak{K}^N_0\oplus\cdots \oplus\mathfrak{K}^N_N$$
where $\mathfrak{K}^N_0={\rm span}(u_0\otimes\cdots\otimes u_0)$ and 
$$\mathfrak{K}^N_k=\underbrace{u_0\otimes\cdots\otimes u_0}_{N-k}\otimes_s \bigotimes_\text{sym}^k\mathfrak{H}_+=u_0^{\otimes (N-k)}\otimes_s\gH_+^k,$$
where we recall again that $\mathfrak{H}_+=\{u_0\}^\perp={\rm span}\{u_1,u_2,...\}\subset L^2(\Omega)$.
In other words, we can write any wavefunction $\Psi\in \H^N$ as follows
$$\Psi:=\psi_0\, u_0^{\otimes N}+ u_0^{\otimes (N-1)}\otimes_s\psi_1 + u_0^{\otimes (N-2)}\otimes_s\psi_2+\cdots + \psi_N$$
where $\psi_k\in\gH_+^k$.
It is a simple exercise to verify that
$$\pscal{u_0^{\otimes (N-k)}\otimes_s\psi_k,u_0^{\otimes (N-\ell)}\otimes_s\psi_\ell}_{\gH^N}=\pscal{\psi_k,\psi_\ell}_{\gH^k}\,\delta_{k\ell}$$
from which we deduce that 
$$\|\Psi\|^2=|\psi_0|^2+\sum_{k=1}^N\norm{\psi_k}_{\gH^k_+}^2.$$
Therefore we see that the linear map
\bq \label{eq:def-unitary-UN}
\begin{array}{cccl}
U_N: & \gH^N & \to & \displaystyle\cF_+^{\leq N}=\bigoplus_{n=0}^N\gH_+^n\\[0.3cm]
 & \Psi & \mapsto & \psi_0\oplus \psi_1 \oplus\cdots \oplus \psi_N
  \end{array}
  \eq
is a unitary operator from $\gH^N$ onto the truncated Fock space $\cF_+^{\leq N}$. The latter can always be seen as being embedded in the full Fock space $\cF^+$ of excited particles and the unitary operator $U_N$ is also a partial isometry from $\gH^N$ to $\cF^+$. We see that, in the limit $N\to\ii$, the Fock space $\cF_+$ of excited particles arises naturally as the limit of the truncated Fock spaces $\cF_+^{\leq N}$.

The operator $U_N$ is a mathematical tool which implements what is called a $c$-number substitution~\cite{LieSeiYng-05,LieSeiSolYng-05}. In Fock space the usual way to formalize the $c$-number substitution is to use the Weyl operator, and here $U_N$ plays the same role. The difference is that the Weyl operator is defined on the whole Fock space $\cF$ with values in $\cF$, while with the operator $U_N$ we go immediately from the $N$-body space $\gH^N$ to the excitation Fock space $\F_+$, which is a proper subspace of $\cF$. We shall give important properties of the operator $U_N$ in Section \ref{sec:operators}.

One of our main results will be that 
$$U_N\big(H_N-N\,e_{\rm H}\big)U_N^*\to\bH$$
in an appropriate (weak) sense.

\subsection{Convergence of the excitation spectrum} \label{subsec:excitation-spectrum}

A convenient way to describe the lower eigenvalues is to use the min-max principle (see \cite{ReeSim4}). If $\A$ is a self-adjoint operator, which is bounded from below, on a (separable) Hilbert space, then we may define the min-max values
$$ \lambda_L (\A):= \inf_{\substack{Y~ {\rm subspace}\\{\rm dim}\,Y=L}} \,\,\,\max_{\substack{\Phi\in Y\\ ||\Phi||=1}} \langle \Phi, \A \Phi \rangle$$
for $L=1,2,...$. It is known that $\lim_{L\to \infty}\lambda_L(\A) = \inf \sigma_{\rm ess}(\A)$, where we use the convention that $\inf \sigma_{\rm ess}(\A)=+\infty$ when the essential spectrum of $\A$ is empty. Moreover, if $$\lambda_L(\A)< \inf \sigma_{\rm ess}(\A),$$
then $\{\lambda_j(\A)\}_{j=1}^L$ are the lowest $L$  eigenvalues of $\A$, counted with multiplicity.

Our main result is the following.

\begin{theorem}[Convergence of the excitation spectrum]\label{thm:lower-spectrum} Assume that {\rm (A1)-(A2)-(A3)} hold true. 

\smallskip

\noindent $(i)$ \emph{(Weak convergence to $\bH$)}. For every fixed $\Phi$ and $\Phi'$ in the quadratic form domain of the Bogoliubov Hamiltonian $\bH$, we have 
\begin{equation}
\boxed{\lim_{N\to\ii}\pscal{\Phi',U_N\big(H_N-N\,e_{\rm H}\big)U_N^*\,\Phi}_{\cF_+}=\pscal{\Phi',\bH\,\Phi}_{\cF_+}}
\label{eq:weak-CV}
\end{equation}
where $U_N$ is defined in (\ref{eq:def-unitary-UN}) and by convention $U_N^*$ is extended to 0 outside of $\cF_+^{\leq N}$.

\smallskip

\noindent $(ii)$ \emph{(Convergence of eigenvalues).} Let $\lambda_1(H_N)\le \lambda_2(H_N) \le ...$ and $\lambda_1(\mathbb{H})\le \lambda_2(\mathbb{H}) \le ...$ be the min-max values of $H_N$ in $\H^N$ and $\mathbb{H}$ in $\F_+$, respectively. 
We have 
$$\boxed{\lim_{N\to \infty}\Big( \lambda_L(H_N)- Ne_{\rm H} \Big) =\lambda_L (\mathbb{H})}$$
for every $L=1,2,...$. Consequently, we have the spectral gap
$$ \liminf_{N\to \infty}\!\Big( \inf \sigma_{\rm ess} (H_N) - \lambda_1(H_N) \Big) \ge \inf \sigma_{\rm ess}(\mathbb{H})- \lambda_1(\mathbb{H})= \inf \sigma_{\rm ess}(h) >0. $$

\smallskip

\noindent $(iii)$ \emph{(Convergence of the ground state)}. The lowest eigenvalue of $\bH$ is simple, with corresponding ground state 
$\Phi^{(1)}$ in $\cF_+$ (defined up to a phase factor). Hence the lowest 
eigenvalue of $H_N$ is also simple for $N$ large enough, with ground state 
$\Psi_N^{(1)}$. Furthermore (up to a correct choice of phase for $\Psi^{(1)}_N$), 
\begin{equation}
\boxed{\lim_{N\to \infty} U_N \Psi_N^{(1)} = \Phi^{(1)} }
\label{eq:CV_gd_state}
\end{equation}
strongly in $\cF_+$. The latter convergence is strong in the norm induced by the quadratic form of $\bH$ on $\cF_+$ if {\rm (A3s)} holds true.

\smallskip

\noindent $(iv)$ \emph{(Convergence of lower eigenvectors)}.  Assume that $\lambda_L(\mathbb{H})<\inf \sigma_{\rm ess}(\mathbb{H})$ for some $L\ge 1$. Then $\lambda_L(\mathbb{H})$ is the $L^{th}$ eigenvalue of $\mathbb{H}$ and for $N$ large enough, $\lambda_L(H_N)$ is the $L^{th}$ eigenvalue of $H_N$. Furthermore,  if $\big(\Psi_N^{(L)}\big)_{N\geq2}$ is a sequence of associated eigenvectors, then, up to a subsequence, 
\begin{equation}
\boxed{\lim_{N\to \infty} U_N \Psi_N^{(L)} = \Phi^{(L)}}
\label{eq:CV_states}
\end{equation}
strongly in $\cF_+$, where $\Phi^{(L)}$ is an eigenvector of $\mathbb{H}$ associated with the eigenvalue $\lambda_L(\mathbb{H})$. The latter convergence is strong in the norm induced by the quadratic form of $\bH$ on $\cF_+$ if {\rm (A3s)} holds true.
\end{theorem}

The proof of Theorem~\ref{thm:lower-spectrum} is provided in Section~\ref{sec:proof-Theorem-lower-spectrum}. Let us now make some comments on this result. 

The weak limit~\eqref{eq:weak-CV} shows how the Bogoliubov Hamiltonian $\bH$ arises from the particle-conserving Hamiltonian $H_N$. This convergence essentially implies the upper bounds on the eigenvalues of $H_N$ and this is the easy part of our proof. In previous results (for instance in~\cite{Solovej-06}) the upper bounds were more involved because the argument was carried out in the original Fock space $\cF$ and the trial state had to be projected on $\gH^N$. On the contrary we work in the excited Fock space $\cF_+$ and only need to project the state on $\cF^{\leq M}_+$ with $M\leq N$ (in practice $1\ll M\ll N$) before applying $U_N^*$, which is much easier. Note that we actually do \emph{not} need the condensation assumption (A3) for showing $(i)$ and the upper bound on the eigenvalues.

The difficult part of the proof is the lower bound on the eigenvalues, which requires (A3) as well as a localization method in the Fock space $\cF_+$, in the spirit of a previous result of Lieb and Solovej~\cite[Thm. A.1]{LieSol-01}. The idea is to estimate the error made by replacing a vector $\Phi\in\cF_+$ by its truncation on $\cF_+^{\leq M}$, in a lower bound. This method is the object of Section~\ref{sec:localization} where we prove an IMS-type localization formula.

Even if our simplified statement does not reflect this, we are able to prove explicit error estimates. For instance, for the convergence of $\lambda_L(H_N)$ in $(ii)$, we will prove that 
\begin{equation}
 -C_L\big(\eps_{R_L}(N)+N^{-1/3}\big)\le \lambda_L(H_N)-Ne_{\rm H} - \lambda_L (\mathbb{H}) \le C_L N^{-1/3}
\label{eq:speed} 
\end{equation}
for $N$ large enough, where $C_L$ and $R_L$ only depend on $(\lambda_L(\mathbb{H})-\lambda_1(\mathbb{H}))$ and 
$(\lambda_L({H_N})-\lambda_1({H_N}))$, and where we recall that $\eps_{R_L}$ is given in Assumption {\rm (A3)}. Similar estimates can be obtained for the eigenfunctions. Under the strong condensation assumption (A3s), our proof shows that one can take $\varepsilon_R(N)=O(N^{-1})$, leading to an overall error of the order $O(N^{-1/3})$.

In~\cite{Seiringer-11,GreSei-12}, Seiringer and Grech have shown a result similar to Theorem~\ref{thm:lower-spectrum}. More precisely, in~\cite{Seiringer-11} Seiringer treated the case of $\Omega$ a cube in any dimension, $T=-\Delta$ with periodic boundary conditions, and $w$ a bounded and positive periodic function such that $\widehat{w}\geq0$. In~\cite{GreSei-12}, the same method was used to treat the case of $\Omega=\R^d$, $T=-\Delta+V_{\rm ext}$ with $V_{\rm ext}\to\ii$ at infinity, and $w$ a bounded  positive function such that $\widehat w\geq0$. In these two cases, the properties (A1), (A2) and (A3s) are all satisfied and Theorem~\ref{thm:lower-spectrum} applies. The proof of strong condensation (A3s) is simple and relies on $\widehat{w}$ being positive, see, e.g.,~\cite[Lemma 1]{Seiringer-11}. 

The speed of convergence~\eqref{eq:speed} that we can establish in our abstract setting is slightly worse than the $O(N^{-1/2})$ obtained in~\cite{Seiringer-11,GreSei-12}. The method of proof is rather different, however. In~\cite{Seiringer-11,GreSei-12}, the authors relate $H_N$ to an $N$-dependent Hamiltonian $H^{\rm Bog}$ which is quadratic in the effective annihilation operators $b_j=a_ja_0^*/\sqrt{N-1}$. These operators only satisfy the commutation relations in the limit $N\to\ii$. The effective Bogoliubov Hamiltonian $H^{\rm Bog}$ can be diagonalized by a unitary transform, up to an error. The unitary is constructed by inserting the effective $b_j$'s in the formula of the Bogoliubov rotation which diagonalizes the true Hamiltonian $\bH$. 

In the present paper, by applying the unitary $U_N$, we settle the problem in the excited Fock space $\cF_+$, in which the true Bogoliubov Hamiltonian $\bH$ acts. The main advantage of our approach is that $\cF_+$ and $\bH$ are now independent of $N$, which makes the comparison with $U_NH_NU_N^*$ clearer, in our opinion. The effective operators $b_j$'s were also used in previous works~\cite{LieSol-01,LieSol-04,LieSeiSolYng-05} on the one- and two-component Bose gases, for which our approach could be useful as well. 

Let us remark that the convergence~\eqref{eq:CV_gd_state} in the Fock space $\cF_+$ can be rewritten in the original $N$-body space $\gH^N$ as
\begin{equation*}
\norm{\Psi_N^{(1)}-U_N^*\Phi^{(1)}}_{\gH^N}=\Big\|\Psi_N^{(1)}-\phi_0^{(1)}\,u_0^{\otimes N}-\sum_{j=1}^N\phi_j^{(1)}\otimes_s u_0^{\otimes (N-j)}\Big\|_{\gH^N}\underset{N\to\ii}{\longrightarrow}0
\end{equation*}
where $\Phi^{(1)}=\phi_0^{(1)}\oplus\phi_1^{(1)}\oplus\cdots\in\cF_+$. In particular, when $w\neq0$, we see that the many-body ground state $\Psi_N^{(1)}$ is \emph{never} close to the Hartree state $u_0^{\otimes N}$ in the norm of $\gH^N$. This is because the vacuum in $\cF_+$ is never an eigenvector of the Bogoliubov Hamiltonian $\bH$, except when $w\equiv0$. A similar property holds for the lower excited states $\Psi_N^{(L)}$.

Although the wavefunction $\Psi_N^{(L)}$ is in general not close to the Hartree state $u_0^{\otimes N}$ in the norm of $\gH^N$, its density matrices are close to that of $u_0^{\otimes N}$. Indeed, if (A1)-(A2)-(A3s) hold, then the convergence of eigenvectors  in $(iv)$ implies the following convergence of the one-body density matrices
$$ \lim_{N\to \infty}\Tr \Big| Q \gamma_{\Psi_N^{(L)}} Q - \gamma_{\Phi^{(L)}} \Big|=0$$
and
$$\lim_{N\to \infty} \Tr \Big| P \gamma_{\Psi_N^{(L)}} P - (N-N_+^{(L)} ) |u_0 \rangle \langle u_0 | \Big|= 0$$
where $P=|u_0 \rangle \langle u_0 |=1-Q$ and $ N_+^{(L)}=\Tr \gamma_{\Phi^{(L)}}$. In particular, by simply controlling the cross term $P\gamma_{N}^{(L)}Q$ by the Cauchy-Schwarz inequality, we immediately obtain the complete Bose-Einstein condensation
$$
\Tr \Big| N^{-1}\gamma_{\Psi_N^{(L)}}- |u_0 \rangle \langle u_0| \Big| \le O(N^{-1/2}).
$$ 
More generally, we can show that if $\Psi$ is a wave function in $\gH^N$ satisfying $\langle \Psi_N, H_N \Psi_N\rangle \le E(N)+O(1)$, then we have 
\bq \label{eq:cv-density-matrices}
\Tr \Big| \Gamma^{(k)}_{\Psi_N}- |u_0^{\otimes k} \rangle \langle u_0^{\otimes k}| \Big| \le O(N^{-1/2})
\eq
where $k\in \mathbb{N}$ is fixed, and $\Gamma_{\Psi_N}^{(k)}$ is the operator on $\gH^{k}$ with kernel 
\[\begin{gathered}
   \Gamma_{\Psi_N}^{(k)}(x_1,...,x_k; y_1,...,y_k)\hfill \\
   \quad \quad = \int \Psi_{N} (x_1,...,x_N) \overline{\Psi_N (y_1,...,y_N)}\, \d x_{k+1}... \d x_{N}\, \d y_{k+1}... \d y_{N}.\hfill \\ 
\end{gathered} \]
Note that, by looking at the first order of the density matrices in~\eqref{eq:cv-density-matrices}, the excited states cannot be distinguished from the ground state.
A slightly weaker version of (\ref{eq:cv-density-matrices}), namely that $\Gamma^{(k)}_{N}(x_1,...,x_k; x_1,...,x_k)$ converges to $|u_0(x_1)|^2... |u_0(x_k)|^2$ weakly in $\gH^k$, was proved recently by Kiessling for the ground state of bosonic atoms \cite{Kiessling-12}. In fact, the convergence of reduced density matrices in (\ref{eq:cv-density-matrices}) is well understood in the time-dependent setting~\cite{Spohn-80, ErdYau-01, AmmNie-08, AmmNie-09, RodSch-09, CheLeeSch-11}, and in this case there has been recent interest in corrections to the Hartree equation~\cite{GriMacMar-10,GriMacMar-11,BenKirSch-11}, where our method might also apply.

\begin{remark}\it 
In the convergence in (iii), the sequence of ground states $\{\Psi_N^{(1)}\}_{N}$ can be replaced by any sequence of approximate ground states. More precisely, if for every $N$ we take a wave function $\Psi_N'\in \H^N$ such that
$$\lim_{N\to \infty} \left( \langle H_N \rangle _{\Psi_N'} -\lambda_1(H_N) \right) =0, $$
then we still have $U_N \Psi_N' \to \Phi^{(1)}$
in the same sense as in Theorem~\ref{thm:lower-spectrum}.
\end{remark}

\begin{remark}\label{rmk:coupling-constant} \it If we perturbate the factor $1/(N-1)$ in front of the interaction term in the Hamiltonian $H_N$ by a term of $O(N^{-2})$ order, then our results in Theorem \ref{thm:lower-spectrum} remain valid, with the modification that the eigenvalues are shifted by an extra term of order $O(1)$. More precisely, if $\kappa_N= (N-1)^{-1} + \kappa N^{-2} + o(N^{-2})$ with $\kappa\in \mathbb{R}$ fixed, then the min-max values of the Hamiltonian
$$ H_{N,\kappa}:=\sum_{i=1}^N T_i + \kappa_N \sum_{1\le i<j\le N} w(x_i-x_j)$$
satisfy
$$ \lambda_L(H_{N,\kappa})= N e_{\rm H}+\kappa(\mu_{\rm H} - e_{\rm H}) + \lambda_L(\bH)+o(1).$$
\end{remark}

\subsection{Positive temperature}

At a positive temperature $\beta^{-1}>0$, the free energy of the system is given by
$$ F_\beta(N):= \inf_{\substack{\Gamma \ge 0 \\ \Tr_{\gH^N} (\Gamma)=1}} \Big\{ \Tr[H_N \Gamma]-\beta^{-1}S(\Gamma) \Big\}=-\beta^{-1}\log \Tr_{\gH^N} e^{-\beta H_N}$$
where $S(\Gamma):= -\Tr [\Gamma \log \Gamma]$ stands for the von Neumann entropy. Our goal is to establish the convergence
$$ \lim_{N\to\ii}\left(F_\beta(N)- Ne_{\rm H}\right) = -\beta^{-1} \log \Tr_{\F_+} \left[ e^{-\beta \mathbb{H}} \right].$$
For this we need two new conditions. At positive temperature the stability of the system does not follow from the simple relative boundedness assumptions~(A1) on $w$. So we need the following

\medskip

\noindent (A4) (Stability). {\it There exists $\beta_0>0$ such that $F_{\beta_0}(N)\ge -CN$ for all $N$ and $\Tr_{\F_+} \left[ e^{-\beta_0 \bH } \right] <\infty$.}

\medskip

Our second new assumption is a modified version of the zero-temperature condensation~(A3), which we now only assume to hold for the Gibbs state at temperature $\beta^{-1}$ for simplicity.

\medskip

\noindent (A3') (Bose-Einstein condensation at positive temperature). {\it For any $\beta^{-1}<\beta_0^{-1}$, one has
\begin{equation}
\lim_{N\to \infty}\frac{\left\langle { u_0,  \gamma_{\beta,N} u_0} \right\rangle}{N} = 1
\label{eq:condensation-positive-temperature}
\end{equation}
where the one-body density matrix of the Gibbs state $\Gamma_{\beta,N}:=e^{-\beta H_N}$ $/\Tr\left[e^{-\beta H_N}\right]$ is designated by $\gamma_{\beta,N}$, namely, in terms of kernels, 
\[
\gamma_{\beta ,N} (x,y) := N \int\limits_{\Omega^{N-1}} \Gamma_{ \beta, N} (x,x_2,...,x_N; y,x_2,...,x_N) d{x_2}...d{x_N}.
\]
}

Let us remark that if the strong condensation assumption (A3s) holds true for some $\eps_0\in (0,1)$ and $\Tr_{\F_+} \left[ e^{-(1-\eps_0)\beta_0\, \bH } \right] <\infty$ for some $\beta_0$, then we can prove (A3') and (A4) for the corresponding $\beta_0$. We of course always assume that (A1) and (A2) hold true. Moreover, if $h\ge \eta_{\rm H}$ and $K_1=K_2\ge 0$, then 
$${\rm d}\Gamma(h+C)+C \ge \bH \ge {\rm d}\Gamma(h-\eps)-C_\eps$$
(see (\ref{eq:comparison-bH-dGamma(h+1)-general}) in Appendix \ref{apd:Bogoliubov-Hamiltonian}), and hence the condition $\Tr_{\F_+} \left[ e^{-(1-\eps_0)\beta_0 \bH } \right] <\infty$ is equivalent to $\Tr \left[ e^{-(1-\eps_0)\beta_0 h } \right] <\infty$. The latter holds true if we have  $\Tr \left[ e^{-(1-\eps_0)(1-\alpha_1)\beta_0 T } \right] <\infty$, because $h\ge (1-\alpha_1)T-C$, where $\alpha_1\in(0,1)$ is given in Assumption (A1). 

Our main result is the following

\begin{theorem} [Positive temperature case] \label{thm:positive-temperature} Assume that {\rm (A1)-(A2)-(A3')-(A4)} hold true. Then for every $\beta^{-1}<\beta_0^{-1}$, we have
$$\boxed{\lim_{n\to\ii}\Tr_{\F_+}\left| U_N e^{-\beta (H_N-Ne_{\rm H})} U_N^* - e^{-\beta \bH}\right|=0.}$$
This implies the convergence of the corresponding Gibbs states and of the free energy:
$$ \lim_{N\to \infty} \left( F_\beta(N) - Ne_{\rm H} \right) =- \beta^{-1} \log \Tr_{\F_+} \left[ e^{-\beta \mathbb{H}} \right].$$
\end{theorem}

Theorem \ref{thm:positive-temperature} is proved using the same argument as that of the proof of Theorem \ref{thm:lower-spectrum}, together with a well-known localization inequality for the entropy, see Section~\ref{sec:proof:positive-temp}.

\section{Applications}\label{sec:applications}

It is easy to verify that the models considered by Seiringer in~\cite{Seiringer-11} and by Grech-Seiringer in~\cite{GreSei-12} satisfy our assumptions (A1), (A2) and (A3s). So our approach applies and we are able to recover their results. 

In the rest of this section, we consider two Coulomb systems for which we are able to apply our main results.

\subsection{Bosonic atoms}\label{sec:bosonic-atoms}

For a bosonic atom we mean a system including a classical nucleus at the origin in $\R^3$ of charge $Z>0$ and  $N$ ``bosonic quantum electrons" of charge $-1$. The system is described by the Hamiltonian 
\[
 \sum\limits_{i = 1}^N {\left( { - \Delta _i  - \frac{Z}
{{|x_i|}}} \right)}  + \sum\limits_{1 \leqslant i < j \leqslant N} {\frac{1}
{{|x_i  - x_j |}}} 
\]
acting on the symmetric space $\H^N=\bigotimes_{\text{sym}}^N L^2 
(\mathbb{R}^3 )$. For simplicity of writing, we only consider spinless electrons. We shall study the asymptotics of $H_{N,Z}$ when $N\to \infty$ and $(N-1)/Z$ is fixed.
 
By using the unitary $\mathcal{D}_{\ell}:\H^N \to \H^N$ defined by $ (\mathcal{D}_{\ell} \Psi) (x_1,...,x_N)=\ell^{3N/2}\Psi(\ell x_1, ..., \ell x_N)$ with $\ell=N-1$, we can rescale the Hamiltonian to  
\[
H_{t,N}:= \sum\limits_{i = 1}^N {\left( { - \Delta _i  - \frac{1}
{{t |x_i|}}} \right)}  + \frac{1}{N-1}\sum\limits_{1 \leqslant i < j \leqslant N} {\frac{1}
{{|x_i  - x_j |}}}
\]
where $t:=(N-1)/Z$. The Hamiltonian $H_{t,N}$ has the same form as in the previous section, with $\Omega=\R^3$, $T=-\Delta-1/(t|x|)$ and $w(x-y)=|x-y|^{-1}$. The eigenvalues of the original Hamiltonian are then proportional to $(N-1)^2$ times the eigenvalues of $H_{t,N}$.

Note that Assumption (A1) holds due to Kato's inequality
\bq \label{eq:Kato-inequality}
 \frac{1}{|x|} \le -\eps \Delta_{x} + \frac{1}{4\eps}\,\,\,\,\,\,~~{\rm on}~L^2(\R^3).
 \eq
and the fact that $H_{t,N}$ is self-adjoint on $\bigotimes_{\rm sym}^{N} H^2(\R^2)$. Note also that the first eigenvalue of $H_{t,N}$, when it exists, is always non-degenerate. In the following we shall denote by $\Phi_{t,N}^{(1)}$ the corresponding unique positive ground state.

It was already proved by Benguria and Lieb \cite{BenLie-83} that for every $t>0$, the leading term of the ground state energy of $H_{t,N}$ is given by Hartree's energy, that is  
\bq \label{eq:Leading-term-bosonic-atoms}
\boxed{\inf \sigma ({H_{t,N}})=N e_{\rm H}(t)+o(N)}
\eq
as $N\to \infty$, where
\bq \label{eq:Hartree-energy-bosonic-atoms}
e_{\rm H}(t) := \inf_{\substack{u \in {L^2}({\R^3})\\||u|| = 1}} \left\{ {\int_{\R^3}|\nabla u(x)|^2\,\d x-\frac1t\!\int_{\R^3}\frac{|u(x)|^2}{|x|}\,\d x  + \frac{1}{2}D(|u|^2,|u|^2)} \right\}.
\eq
Here, again
$$ D(f,g): = \iint\limits_{{\mathbb{R}^3} \times {\mathbb{R}^3}} \frac{\overline {f(x)} g(y)}{|x-y|}\,\d x\,\d y.$$
Note that $D(f,f)\ge 0$ because the Coulomb potential $|x|^{-1}$ has positive Fourier transform.

The fact that the next order in the expansion of the ground state energy of bosonic atoms is given by Bogoliubov's theory was first conjectured in~\cite{Nam-thesis}. In the following, by applying Theorem \ref{thm:lower-spectrum},  we shall establish not only this conjecture but also many other properties of the system.

By a convexity argument, it can be shown (see \cite{Lieb-81}) that the Hartree minimization problem (\ref{eq:Hartree-energy-bosonic-atoms}) has a minimizer $u_t$ if and only if $t\le t_c$, for some critical number $t_c \in (1,2)$ (it was numerically computed in \cite{Baumgartner-84} that $t_c\approx 1.21$). In the case of existence, the minimizer is unique, positive, radially-symmetric. Moreover it decays exponentially and it solves the mean-field equation
$$
\begin{cases}
h_t\,u_t=0,\\[0.2cm]
h_t:=- \Delta  - \dfrac{1}
{{t|x|}} + |u_t|^2 *\frac{1}
{{|x|}} -\mu_{\rm H}(t) ,
\end{cases}
$$
with the Lagrange multiplier $\mu_{\rm H}(t) \le 0$. Moreover, if $t<t_c$, then $\mu_{\rm H}(t)<0$ and there is a constant $\eta_{\rm H}(t)>0$ such that  
\bq \label{eq:gap-h-bosonic-atoms}
h_t \ge \eta_{\rm H}(t) >0~~{\rm on}~\gH_+ := \{u_t\}^{\bot}.
\eq

The critical binding number $t_c$ in Hartree's theory also plays an important role for the original quantum problem. In fact, it was shown in \cite{BenLie-83,Solovej-90,Bach-91} that for every $N$ there are two numbers $b(N)\le b'(N)$ satisfying that $H_{t,Z}$ always has a ground state if $t\le b(N)$ and $H_{t,Z}$ has no ground states if $t\ge b'(N)$, and that
$$\lim_{N\to \infty} b(N)= \lim_{N\to \infty} b'(N)=t_c.$$

In the following we shall always assume that $t$ is fixed strictly below $t_c$. In this case, Assumption (A2) holds true. In fact, due to Hardy's inequality
\bq \label{eq:Hardy-inequality}
\frac{1}{4|x|^2} \le -\Delta_x \,\,\,\,\,\,~~{\rm on}~L^2(\R^3),
\eq
the function 
$ K_t(x,y):=u_t(x)|x-y|^{-1}u_t(y)$
belongs to $L^2((\R^3)^2)$. Hence, $K_t(x,y)$ is the integral kernel of a Hilbert-Schmidt operator, still denoted by $K_t$.  Note that $K_t\ge 0$ because $|x-y|^{-1}$ is a positive kernel. Thus the spectral gap (\ref{eq:gap-h-bosonic-atoms}) implies the non-degeneracy of the Hessian, namely
\bqq 
\left( {\begin{array}{*{20}{c}}
  {h_t + K_t}&K_t \\ 
  K_t&{{{h_t}} + K_t} 
\end{array}} \right) \geqslant \eta_{\rm H}(t)\qquad \text{on}\quad \gH_+\oplus \gH_+ .
\eqq

The condensation in Assumption (A3) is implicitly contained in the proof of the asymptotic formula (\ref{eq:Leading-term-bosonic-atoms}) by Benguria and Lieb. In fact, the upper bound in (\ref{eq:Leading-term-bosonic-atoms}) can be seen easily by using the Hartree state $u_t^{\otimes N}$. The lower bound is more involved and it follows from the Lieb-Oxford inequality \cite{Lieb-79,LieOxf-80} which says that for every wave function $\Psi\in \H^N$,
\bq \label{eq:Lieb-Oxford-inequality}
\left \langle \Psi, \left(\sum_{1\le i<j\le N}\frac{1}{|x_i-x_j|} \right)\Psi\right \rangle \ge \frac{1}{2} D(\rho_\Psi, \rho_\Psi ) -1.68 \int_{\R^3} \rho_\Psi (x)^{4/3}\d x .
\eq
The following is a quantitative version of the strong condensation (A3s).

\begin{lemma}[Strong condensation of bosonic atoms] \label{eq:condensation-bosonic-atoms} If $t\le t_c$, then
\begin{equation}
 H_{t,N}-N e_{\rm H} \ge \left( 1-  N^{-2/3}\right)  \sum_{i=1}^N (h_t)_i - CN^{1/3}. 
\label{eq:strong-condensation-Bose-gas} 
\end{equation}
\end{lemma}

\begin{remark}\it In particular, from (\ref{eq:strong-condensation-Bose-gas}) it follows that if $\langle H_{t,N}\rangle_\Psi \le N e_{\rm H} + R$, then $\langle u_0,\gamma_\Psi u_0\rangle \ge N-CN^{1/3}$ which is (A3). When $t<t_c$, by Theorem~\ref{thm:lower-spectrum} we can improve the estimate to $\langle u_0,\gamma_\Psi u_0\rangle \ge N-C$. The latter was shown by Bach, Lewis, Lieb and Siedentop in \cite{BacLewLieSie-93} using a different method. 
\end{remark}

\begin{proof} We start by estimating the terms on the right side of the Lieb-Oxford inequality (\ref{eq:Lieb-Oxford-inequality}). First, from the positivity $D(f,f)\ge 0$, we have
\bq \label{eq:proof-condensation-atoms-eq1}
\frac{1}{2} D(\rho_\Psi, \rho_\Psi) &\ge & D(\rho_\Psi, N|u_t|^2)-\frac{1}{2}  D(N|u_t|^2,N|u_t|^2)\nn\hfill\\
&=& N D(\rho_\Psi, |u_t|^2)+ N^2(e_{\rm H}(t)-\mu_{\rm H}(t)).
\eq
On the other hand,  using the Hoffmann-Ostenhof inequality~\cite{Hof-77} and Sobolev's inequality \cite[Theorem 8.3]{LieLos-01}  we can estimate 
$$ \left \langle \Psi, \left( \sum_{i=1}^N -\Delta_i \right) \Psi \right \rangle = \Tr[-\Delta \gamma_\Psi] \ge \left \langle \sqrt{\rho_\Psi}, -\Delta \sqrt{\rho_\Psi} \right \rangle \ge C \left( \int_{\R^3} \rho_\Psi^3 \right)^{1/3}.$$
Therefore, by H\"older inequality, we find that
\bqq 
\int \rho_{\Psi}^{4/3} \le \left( \int \rho_\Psi ^3 \right) ^{1/6} \left( \int \rho_\Psi \right) ^{5/6} \le \eps \Tr[-\Delta \gamma_\Psi] + C\frac{N^{5/3}}\eps .
\eqq
Thus using the Lieb-Oxford inequality (\ref{eq:Lieb-Oxford-inequality}), the  estimate (\ref{eq:proof-condensation-atoms-eq1}) and $h_t \ge -\Delta/2 -C$ we get
$$
\langle H_{t,N}\rangle_\Psi  \ge   N e_{\rm H}(t) + \left( 1-\frac{2\eps}{N} \right) \Tr[h_t \gamma_\Psi] - 2C\eps  - C\eps^{-1} N^{2/3}
$$
for all $\eps>0$. Replacing $\eps$ by $ N^{1/3}/2$, we obtain~\eqref{eq:strong-condensation-Bose-gas}.
\end{proof}

All this shows that if $t<t_c$, then Assumptions (A1)-(A2)-(A3s) hold true, and we may apply Theorem \ref{thm:lower-spectrum} to show that the lower spectrum of $H_{t,Z}$ converges to the lower spectrum of the Bogoliubov Hamiltonian 
\begin{multline*}
\mathbb{H}_t:=\int_{\Omega}a^*(x)\big( (h_t+K_t)\,a\big)(x)\,\d x\\+\frac12\int_{\Omega}\int_{\Omega} K_t(x,y)\bigg( a^*(x)a^*(y)+ a(x)a(y)\bigg)\d x\,\d y, 
\end{multline*}
which acts on the Fock space
$\F_+ = \bigoplus_{n=0}^{\infty} \bigotimes_{\rm sym}^{n} \gH_+$.
Beside some basic properties of $\mathbb{H}_t$ already given in Theorem \ref{thm:Bogoliubov-Hamiltonian}, we have the following additional information.

\begin{proposition}[Bogoliubov Hamiltonian of bosonic atoms] For every $t\in (0,t_c)$ one has
$$ \sigma_{\rm ess} (\mathbb{H}_{t})= \left[ \inf \sigma ( \mathbb{H}_ t)  -\mu_{\rm H}(t),\infty \right).
$$
Moreover, if $t<1$, then $\mathbb{H}_t$ has infinitely many eigenvalues below its  essential spectrum. On the other hand, if $t\ge 1$, then  $\mathbb{H}_t$ only has finitely many eigenvalues below its essential spectrum.
\end{proposition}

\begin{proof} First, since $h_t=-\Delta+V_t(x)-\mu_{\rm H}(t)$, where  $V_t(x):=(|u_0|^2*|.|^{-1})(x)-t^{-1}|x|^{-1}$
is relatively compact with respect to $-\Delta$, we obtain
$ \sigma_{\rm ess}(h_t)= [-\mu_{\rm H}(t),\infty)$. Thus $\sigma_{\rm ess} (\mathbb{H}_{t})= \left[ \inf \sigma ( \mathbb{H}_ t)  -\mu_{\rm H}(t),\infty \right)$ by (ii) in Theorem \ref{thm:Bogoliubov-Hamiltonian}. The other statements follow from (iii) in Theorem \ref{thm:Bogoliubov-Hamiltonian} and the fact that $h_t$ has infinitely many eigenvalues below its essential spectrum  if and only if $t<1$. 

In fact, if $t<1$, then $V_t(x)=m_t*|.|^{-1}$ where the measure
$$ m_t:=|u_0|^2 -t^{-1} \delta_0$$
has negative mass $\int m <0$, and we can follow \cite[Lemma II.1]{Lions-87}.

On the other hand, if $t \ge 1$, then by applying Newton's Theorem for the radially symmetric function $|u_0(x)|^2$, we can write 
$$ V_t(x) \ge \int_{\R^3} \frac{|u_t(y)|^2}{\max\{|x|,|y|\}}\d y - \frac{1}{|x|} = \int_{|y|\ge |x|} {|u_t(y)|^2} \left( \frac{1}{|y|}-\frac{1}{|x|} \right)\d y.$$
Because $u_t(x)$ decays exponentially, we obtain that $[V(x)]_-$ also decays exponentially and hence belongs to $L^{3/2}(\R^3)$. Therefore, by the CLR bound (see e.g. \cite[Theorem 4.1]{LieSei-09}), we conclude that $h_t$ has finitely many eigenvalues below $-\mu_{\rm H}(t)$. 
\end{proof}

By the celebrated HVZ Theorem~\cite{ReeSim4}, one has 
$$\sigma_{\ess}(H_{t,N})=[\Sigma_{t,N},\infty)$$
where $\Sigma_{t,N}$ is the ground state energy of the $(N-1)$-body Hamiltonian
$$ \sum\limits_{i = 1}^{N-1} {\left( { - \Delta _i  - \frac{1}
{{t |x_i|}}} \right)}  + \frac{1}{N-1}\sum\limits_{1 \leqslant i < j \leqslant N-1} {\frac{1}
{{|x_i  - x_j |}}}.$$
An asymptotic formula for $\Sigma_{t,N}$ can be obtained by the same method as that of Theorem \ref{thm:lower-spectrum} (see Remark \ref{rmk:coupling-constant} in the end of Sec. \ref{subsec:excitation-spectrum}), namely 
$$\Sigma_{t,N}= Ne_{H}(t)-\mu_{\rm H}(t)+\inf \sigma(\mathbb{H}_t)+o(1)=N e_{\rm H}(t)+\inf \sigma_{\rm ess}(\mathbb{H}_t)+o(1).$$
It thus turns out that the essential spectrum of $H_{t,N}-Ne_{\rm H}(t)$ converges to the essential spectrum of $\mathbb{H}_t$.

Thus, by combining the HVZ Theorem and Theorem \ref{thm:lower-spectrum}, we obtain the convergence of the whole spectrum of $H_{t,N}$ as follows.

\begin{theorem} [Spectrum of bosonic atoms]\label{thm:spectrum-bosonic-atoms} If $t\in (0,t_c)$ is fixed, then we have the following statements in the limit $N\to \infty$.

\smallskip

\noindent $(i)$ \emph{(Essential spectrum)}. The essential spectra of $H_{t,N}$ and $\mathbb{H}_t$ are $$\sigma_{\ess}(H_{t,N})=[\Sigma_{t,N},\infty),~~\,\sigma_{\ess}(\mathbb{H}_{t})=[\inf \sigma (\mathbb{H}_t)-\mu_{\rm H}(t),\infty)$$
and we have the convergence
$$ \lim_{N\to \infty} \left( \Sigma_{t,N} -N e_{\rm H}(t) \right) = \inf \sigma (\mathbb{H}_t)-\mu_{\rm H}(t).$$

\smallskip

\noindent $(ii)$ \emph{(Ground state energy and ground state)}. We have
$$ \lim_{N\to \infty}\left( \inf \sigma(H_{t,N}) -Ne_{\rm H} (t) \right) = \inf \sigma (\mathbb{H}_t)< 0.$$
For $N$ large enough, $H_{t,N}$ has a ground state $\Psi_{t,N}^{(1)}$ (which is then unique and can be assumed positive) and 
$$\lim_{N\to \infty} U_N \Psi_{t,N}^{(1)} = \Phi^{(1)}_t$$
strongly in the norm induced by the kinetic energy in $\cF_+$, where $\Phi_t^{(1)}\in\cF_+$ is the unique ground state of $\bH$ such that $\langle\Phi_t^{(1)},\Omega\rangle_{\cF_+}>0$, $\Omega$ being the vacuum state in $\cF_+$.

\smallskip

\noindent $(iii)$ \emph{(Lower eigenvalues and eigenstates)}. Assume that $\mathbb{H}_t$ has $L$ eigenvalues below its essential spectrum, for some $L\ge 1$. Then for $N$ large enough, $H_{t,N}$ also has $L$ eigenvalues below its essential spectrum, and the $L^{th}$ eigenvalue of $(H_{t,N}-Ne_{\rm H}(t))$ converges to the $L^{th}$ eigenvalue of $\mathbb{H}_t$ as $N\to \infty$. Moreover, if $\Psi_{t,N}^{(L)}$ is a corresponding eigenvector, then up to a subsequence, $U_N \Psi_{t,N}^{(L)}$ converges to an eigenvector of the $L$-th eigenvalue of $\bH_{t}$ as $N\to \infty$.
\end{theorem}

\subsection{Trapped Coulomb gases}\label{sec:2D-log-gas}
We consider a two- or three- dimensional system of $N$ bosons trapped by an external potential and interacting via the Coulomb potential. The system is described by the Hamiltonian 
\[
H_N= \sum\limits_{i = 1}^N (\mathcal{K}_{x_i}+V(x_i)) + \frac{1}{N-1} \sum\limits_{1 \leqslant i < j \leqslant N} {w(x_i-x_j)} 
\]
acting on the symmetric space $\H^N=\bigotimes_{\text{sym}}^N L^2 
(\mathbb{R}^d)$ for $d=2$ or $d=3$, where $\cK=-\Delta$ or $\cK=\sqrt{1-\Delta}-1$ is the kinetic energy operator.

Our study in this section was motivated by a recent work of Sandier and Serfaty~\cite{SanSer-12,SanSer-12b}. They proved in~\cite{SanSer-12b} that for the classical 2D Coulomb systems, the ground state energy behaves like
$$ N e_{\rm H} -\frac{1}{4}\log N + c + o(1) $$  
where $e_{\rm H}$ is the Hartree ground state energy when $\cK\equiv0$ and $c$ is the ground state energy of an infinite system of point charges in the plane, interacting through the Coulomb potential and with a neutralizing jellium background. The $\log N$ term is due to some local scaling invariance (think of replacing $V$ by a constant in a small box and shrinking the system about the center of this box). We shall see that, when the kinetic energy is introduced, since the local scaling invariance is broken and shrinking the particles has a too large cost, the logarithmic term disappears and the next term is of order $O(1)$, given by Bogoliubov's theory. In order to make this vague statement rigorous, we need some natural assumptions on $T=\cK+V$ and on $w$.

\medskip

\noindent {\bf Assumptions.} {\it If $d=2$, then $w(x)=-\log |x|$, $\cK$ is either $-\Delta$ or $\sqrt{-\Delta+1}-1$,  and $V\in L^{2}_{\rm loc} (\R^2,\R)$ satisfying that
\[ \liminf_{|x|\to \infty} \frac{V(x)}{[\log |x|]^2} >0.\]

\smallskip

\noindent If $d=3$, then $w(x)=|x|^{-1}$, $\cK$ is either $-\Delta$ or $\sqrt{-\Delta+1}-1$, and $V\in L^{3}_{\rm loc}(\R^3,\R)$ satisfying that $V(x)\to \infty$ as $|x|\to \infty$.}

\medskip

Let us show that under these conditions, Assumptions (A1)-(A2)-(A3s) in Theorem \ref{thm:lower-spectrum} hold true. First, note that $\cK+V$ is bounded from below on $L^2(\R^d)$ (see \cite[Sec. 11.3]{LieLos-01}). When $d=3$, if $K=\sqrt{1-\Delta}-1$, then the relative bound (\ref{eq:bound-w-1}) in (A1) follows from the relativistic Hardy inequality \cite[Lemma 8.2]{LieSei-09}
$$\frac{2}{\pi |x|}\le \sqrt{-\Delta} \le \sqrt{1-\Delta}\;\;\;\;~~{\rm on}~~L^2(\R^3),$$
while if $\cK=-\Delta$, then we even have the stronger bound 
\bq \label{eq:trapped-gases-bound-w^2}
|w(x-y)|^2 \le C_0[\cK_x +\cK_y +V(x) + V(y)]+C\,\,\,\,~~{\rm on}~L^2((\R^d)^2)
\eq
due to Hardy's inequality (\ref{eq:Hardy-inequality}). When $d=2$, then (\ref{eq:trapped-gases-bound-w^2}) also holds true, due to the estimates $\cK+V \ge C^{-1} [\log (1+ |x|)]^2-C$ and
\bqq
|w(x-y)|^2 \le 2\left[\frac{1}{|x-y|}+ [\log (1+|x|)]^2 + [\log (1+|y|)]^2 \right] + C 
\eqq
for some $C>0$. Thus (A1) holds true.
On the other hand, Assumption (A2) follows from the following   

\begin{proposition}[Hartree theory]\label{le:Hartree-theory-trapped-gases} Under the previous assumptions, the variational problem
\bq \label{eq:Hartree-trapped-gas}
e_{\rm H} :=\inf_{\substack{u \in {L^2}({\R^d})\\||u|| = 1}} \left\{ \left\langle {u,(\cK+V)u} \right\rangle  + \frac{1}{2} \int \int |u(x)|^2 w(x-y) |u(y)|^2 \,\d x\,\d y \right\}
\eq
has a unique minimizer $u_0$ which satisfies that $u_0(x)>0$ for a.e. $x\in \R^d$ and solves the mean-field equation 
$$
\begin{cases}
h\,u_0=0,\\[0.2cm]
h:=\cK+V + |u_0|^2*w-\mu_{\rm H} ,
\end{cases}
$$
for some Lagrange multiplier $\mu_{\rm H}\in \R$. The operator $h$ on $\gH_+=\{u_0\}^\bot$ has only discrete spectrum $0<\lambda_1(h)\le \lambda_2 (h)\le ...$ with $\lim_{i\to \infty} \lambda_i(h)=\infty$. Moreover, the operator $K$ with kernel $K(x,y)=u_0(x)w(x-y)u_0(y)$ is Hilbert-Schmidt on $L^2(\R^d)$ and it is positive on $\gH_+.$ Finally, $u_0$ is non-degenerate in the sense of~\eqref{eq:hK>=eta}.
\end{proposition}

Before proving Proposition \ref{le:Hartree-theory-trapped-gases}, let us mention that in 2D, the Coulomb potential $w(x)=-\log |x|$ does not have positive Fourier transform. More precisely, $\widehat w={\rm pv} ~|\cdot|^{-1}$, the principal value of $|\cdot|^{-1}$. Although $w$ is not a positive type kernel, we still have the following restricted positivity (see \cite{CarLos-92}). 
\begin{proposition}[Coulomb $\log$ kernel]\label{pro:log-kernel} For any $f\in L^1(\R^2) \cap L^{1+\eps}(\R^2)$ for some $\eps>0$ with
$$\int_{\R^2}\log(2+|x|)|f(x)|\,\d x<\ii\qquad\text{and}\qquad \int_{\R^2} f =0,$$ 
we have
$$ 0\leq D(f,f): = - \iint\limits_{{\mathbb{R}^2} \times {\mathbb{R}^2}} {\overline {f(x)} \log |x-y| f(y)}\,\d x\,\d y<\ii.$$
\end{proposition}

We can now provide the 

\begin{proof}[Proof of Proposition \ref{le:Hartree-theory-trapped-gases}] By using the relative bound (\ref{eq:bound-w-1}), Sobolev's embedding \cite[Theorem 8.4]{LieLos-01} and the fact that $V(x)\to \infty$ as $|x|\to \infty$, it is straightforward to show that there is a minimizer $u_0$ for $e_{\rm H}$ in (\ref{eq:Hartree-trapped-gas}) such that $u_0\in D(\cK^{1/2})$ and  
$$ \int_{\R^d} V(x) |u_0(x)|^2 \d x <\infty.$$
Since $\left\langle {u,\cK u} \right\rangle \ge \left\langle {|u|,\cK|u|} \right\rangle$ \cite[Theorems 7.8 and 7.13]{LieLos-01}, we may assume that $u_0 \ge 0$. Moreover, the Hartree equation $hu_0=0$ implies that $u_0(x)>0$ for a.e. $x\in \R^2$ and $u_0$ is then the unique ground state of $h$. Moreover, because $V(x) + (|u_0|^2*w)(x) \to \infty$ as $|x|\to \infty$, the operator $h$ (on $\gH_+$) has only discrete spectrum $\lambda_1(h) \le \lambda_2(h)\le ...$  with $\lambda_1(h)>0$ and $\lim_{i\to \infty} \lambda_i(h)=\infty$ (see \cite[Theorem XIII.16]{ReeSim4}).

For any normalized function $u\in \gH$, we have $\int (|u|^2-|u_0|^2)=0$ and hence $ D(|u|^2-|u_0|^2,|u|^2-|u_0|^2)\ge 0$ by Proposition \ref{pro:log-kernel}. Using this and a convexity argument we can show that $u_0$ is the unique minimizer of $e_{\rm H}$ in (\ref{eq:Hartree-trapped-gas}) .

The operator $K$ with kernel $K(x,y)=u_0(x)w(x-y)u_0(y)$ is Hilbert-Schmidt in $L^2((\R^d)^2)$. When  $d=2$, or $d=3$ and $\K=-\Delta$, this fact holds true due to (\ref{eq:trapped-gases-bound-w^2}). When $d=3$ and $\K=\sqrt{1-\Delta}-1$, the Hilbert-Schmidt property follows from the Hardy-Littlewood-Sobolev inequality \cite[Theorem 4.3]{LieLos-01} and the Sobolev's embedding $H^{1/2}(\R^3)\subset L^3(\R^3)$.  

Finally, the operator $K$ is positive on $\gH_+$ because $\left\langle {v,Kv} \right\rangle = D(vu_0, vu_0) \ge 0$  for every $v\in \gH_+$. The latter inequality follows from Proposition \ref{pro:log-kernel} and the fact that $\int (v u_0)=0$. We then deduce from $h\geq \lambda_1(h)>0$ that $u_0$ is non-degenerate in the sense of~\eqref{eq:hK>=eta}.
\end{proof}

We have shown that (A1) and (A2) hold true. Therefore, we may consider the Bogoliubov Hamiltonian 
\bqq
\mathbb{H}=\!\!\int_{\Omega}\!a^*(x)(h+K)a(x)\,\d x+\frac12\int_{\Omega}\!\!\int_{\Omega} K(x,y)\Big( a^*(x)a^*(y)+ a(x)a(y)\Big)\d x\,\d y,
\eqq
which acts on the Fock space $\F_+ = \bigoplus_{n=0}^{\infty} \bigotimes_{\rm sym}^{n} \gH_+$.
From Theorem \ref{thm:Bogoliubov-Hamiltonian} and the spectral property of $h$, the following can easily be proved.

\begin{proposition}[Bogoliubov Hamiltonian of trapped Coulomb gases]\label{prop:Bogo-Coulomb-gas}
Under the above assumptions on $\cK$, $V$ and $w$, the Bogoliubov Hamiltonian $\bH$ is bounded from below. Its spectrum is purely discrete, consisting of a sequence of eigenvalues $\lambda_1(\mathbb{H}) < \lambda_2(\mathbb{H}) \le \lambda_3(\mathbb{H}) \le \cdots$ with $\lim_{j\to \infty}\lambda_j(\mathbb{H})=\infty$.
\end{proposition}

Now we consider Assumption (A3s). In three dimensions, the condensation can be obtained by following the proof of Lemma \ref{eq:condensation-bosonic-atoms}. In two dimensions, we have the following result. 

\begin{lemma}[Strong condensation of 2D Coulomb gases] Under the above assumptions on $\cK$, $V$ and $w$,  we have 
$$H_{N}-N e_{\rm H} \ge \left( 1-  N^{-1}\right)  \sum_{i=1}^N h_i - \frac{\log N}{4N}-\frac{C}{N}.$$
\end{lemma}

This condensation can be proved by using the same argument as in the proof of Lemma \ref{eq:condensation-bosonic-atoms} and the following Lieb-Oxford type inequality, whose proof can be found in Appendix \ref{apd:log-Lieb-Oxford-inequality}. 

\begin{proposition}[Logarithmic Lieb-Oxford inequality]\label{le:2d_LO}
For any wave function $\Psi\in \bigotimes_1^N L^2(\R^2)$ such that $|\Psi|^2$ is symmetric and $\rho_{\Psi}\in L^1\cap L^{1+\eps}$ for some $0<\epsilon\leq$ and $\int_{\R^2}\log(2+|x|)\rho_\Psi(x)\,\d x<\ii$, we have
\begin{multline}
{\left\langle { \Psi, \left( \sum\limits_{1 \leqslant i < j \leqslant N} {-\log |{x_i} - {x_j}|} \right) \Psi} \right\rangle } \geq  \frac{1}{2} D(\rho_\Psi,\rho_\Psi)-\frac14 N \log N
 \hfill\\
-\frac{C}{\eps}\left(\int_{\R^2}\rho_\Psi\right) -\frac{C}{\epsilon}\left(\int_{\R^2}\rho_\Psi\right)\int_{\R^2}\left(\frac{\rho_\Psi}{\int_{\R^2}\rho_\Psi}\right)^{1+\epsilon}.
\label{eq:2D-Lieb-Oxford} 
\end{multline}
\end{proposition}

The estimate~\eqref{eq:2D-Lieb-Oxford} is probably not optimal but it is sufficient for our purpose. In particular we do not know if the error term involving $\eps$ can be removed.
By applying Theorem ~\ref{thm:lower-spectrum} and Theorem \ref{thm:positive-temperature} we get the following results.

\begin{theorem}[Spectrum of trapped Coulomb gases]\label{thm:trapped-Bose-gases} Assume that $\cK$, $V$ and $w$ satisfy the above assumptions. 

\smallskip

\noindent $(i)$ \emph{(Eigenvalues)}. The Hamiltonian $H_N$ has only discrete spectrum $\lambda_1(H_N)<\lambda_2(H_N) \le \lambda_3(H_N)\le...$ with $\lim_{L\to \infty} \lambda_{L}(H_N)=\infty$. For every $L=1,2,...$ we have the following convergence
$$ \lim_{N\to \infty} \left( \lambda_{L}(H_N) -N e_{\rm H}\right)= \lambda_L(\mathbb{H})$$
where $\lambda_L(\mathbb{H})$  is the $L$-th eigenvalue of the Bogoliubov Hamiltonian $\bH$.

\smallskip

\noindent $(ii)$ \emph{(Eigenvectors)}. If $\Psi_{N}^{(L)}$ is an eigenvector of the $L$-th eigenvalue of $H_{N}$, then up to a subsequence, $U_N \Psi_{N}^{(L)}$ converges (strongly in the norm induced by $\cK+V$) to an eigenvector of the $L$-th eigenvalue of $\mathbb{H}_t$ as $N\to \infty$,  where $U_N$ is defined in (\ref{eq:def-unitary-UN}).

\smallskip

\noindent $(iii)$ \emph{(Positive temperature)}. If we assume furthermore that $\Tr e^{-\beta_0 (\cK+V)}<\infty$ for some $\beta_0>0$, then for every $0<\beta^{-1}<\beta_0^{-1}$ we obtain the convergence
$$\lim_{N\to\ii}\Tr_{\F_+}\left| U_N e^{-\beta(H_N-Ne_{\rm H})} U_N^*- e^{-\beta \bH}\right|=0.$$ 
In particular, we have the convergence of the free energy
\begin{equation}
\lim_{N\to\ii} \Big(-\beta^{-1}\log\Tr_{\gH^N}e^{-\beta H_N}-N\,e_{\rm H}\Big)=-\beta^{-1}\log\Tr_{\cF_+}e^{-\beta \bH}.
\label{eq:Coulomb-gas-free-energy}
\end{equation}
\end{theorem}

We have a couple of remarks on Theorem \ref{thm:trapped-Bose-gases}.

First, we note that the condition $\Tr e^{-\beta_0 (\cK+V)}<\infty$ is satisfied if $V$ grows fast enough at infinity. For example, if $\cK=-\Delta$, $d=2$ or $d=3$, and 
$$\liminf_{|x|\to \infty} \frac{V(x)}{\log |x|} > \frac{d}{\beta_0},$$
then one has the Golden-Thompson-Symanzik inequality \cite{Golden-65,Thompson-65, Symanzik-65} (see also \cite{DolFelLosPat-06} for an elementary proof)
$$\Tr e^{-\beta_0 (\cK+V)}\le (4\pi \beta_0)^{-d/2} \int_{\R^N} e^{-\beta_0 V(x)} dx<\infty.$$
Moreover, if $\cK=\sqrt{-\Delta+1}-1$, $d=2$ or $d=3$, and $$\liminf_{|x|\to \infty} \frac{V(x)}{|x|}>0,$$
then $\Tr e^{-\beta_0(\K+V)}<\infty$ for all $\beta_0>0$, due to \cite[Theorem 1]{DolFelLosPat-06} and the operator inequality $\sqrt{-\Delta}+|x| \ge \sqrt{-\Delta+|x|^2}$ in $L^2(\R^d)$. The latter is a consequence of the operator monotonicity of the square root and the fact that $\sqrt{-\Delta}|x|+|x|\sqrt{-\Delta} \ge 0$ in $L^2(\R^d)$, see \cite[Theorem 1]{HanSie-12}.

Second, in \cite[Theorem 1]{SanSer-12b} the authors provided upper and lower bounds on the classical free energy at nonzero temperatures, which coincide when $\beta^{-1}\to0$. In the quantum case we are able to identify precisely the limit~\eqref{eq:Coulomb-gas-free-energy} even when $\beta^{-1}>0$, which is given by the Bogoliubov Hamiltonian.

\section{Operators on Fock spaces}\label{sec:operators}

In this preliminary section, we introduce some useful operators on Fock spaces and we consider the unitary $U_N$ defined in (\ref{eq:def-unitary-UN}) in detail. 

For any vector $f \in \gH$, we may define the {\it annihilation operator} $a(f)$ and the {\it creation operator} $a^*(f)$ on the Fock space $\F=\bigoplus_{N=0}^\infty \gH^N$ by the following actions  
\bq \label{eq:def-annihilation-operator}
  a(f)\left( {\sum\limits_{\sigma  \in {\mathfrak{S}_N}} {{f_{\sigma (1)}} } \otimes ... \otimes {f_{\sigma (N)}}} \right) = \sqrt N \sum\limits_{\sigma  \in {\mathfrak{S}_N}} {\left\langle {f,{f_{\sigma (1)}}} \right\rangle } {f_{\sigma (2)}} \otimes ... \otimes {f_{\sigma (N)}},
\eq
\bq \label{eq:def-creation-operator}
{a^*}({f_{N }})\left( {\sum\limits_{\sigma  \in {\mathfrak{S}_{N-1}}} {{f_{\sigma (1)}} } \otimes ... \otimes {f_{\sigma (N-1)}}} \right) = \frac{1}{\sqrt {N}} \sum\limits_{\sigma  \in {\mathfrak{S}_N}} {{f_{\sigma (1)}} }  \otimes ... \otimes {f_{\sigma (N )}}
\eq
for all $f,f_1,...,f_{N}$ in $\gH$, and all $N=0,1,2,...$. These operators satisfy the {\it canonical commutation relations}
\bq  \label{eq:CCR}
[a(f),a(g)]=0,\quad[a^*(f),a^*(g)]=0,\quad [a(f),a^*(g)]= \langle f,g \rangle_{\cH}.
\eq

Note that when $f\in \gH_+$, then $a(f)$ and $a^*(f)$ leave $\F_+$ invariant, and hence we use the same notations for annihilation and creation operators on $\F_+$. The operator-valued distributions $a(x)$ and $a^*(x)$ we have used in (\ref{eq:Bogoliubov}) can be defined so that for all $f\in \gH_+$, 
$$ a(f) = \int_\Omega \overline{f(x)} a(x) \d x \quad {\rm and} \quad a^*(f)= \int_\Omega f(x) a^*(x) \d x.$$
To simplify the notation, let us denote $a_n=a(u_n)$ and $a_n^*=a^*(u_n)$, where $\{u_n\}_{n=0}^\infty$ is an orthonormal basis for $L^2(\Omega)$ such that $u_0$ is the Hartree minimizer and $u_n\in  D(h)$ for every $n=1,2,...$. Then the Bogoliubov Hamiltonian defined in (\ref{eq:Bogoliubov}) can be rewritten as
\bq \label{eq:Bogoliubov-Hamiltonian-m-n}
\mathbb{H}= \sum_{m,n\ge 1}  \left\langle u_m,  (h+K_1) u_n \right \rangle_{L^2(\Omega)} a_m^*a_n + \frac{1}{2} \left\langle u_m\otimes u_n, K_2 \right \rangle_{L^2(\Omega^2)}\, a_m^*a_n^*\nn\\+ \frac{1}{2} \left\langle K_2,u_m\otimes u_n \right \rangle_{L^2(\Omega^2)} a_m a_n. \hfill
\eq
The sums here are not convergent in the operator sense. They are well defined as quadratic forms on the domain given in (\ref{eq:core-quadratic-Bogoliubov}). Since the so-obtained operator is bounded from below (by Theorem~\ref{thm:Bogoliubov-Hamiltonian}), it can then be properly defined as a self-adjoint operator by the Friedrichs extension.

It is also useful to lift operators on $\gH^N$ to the Fock space $\F$. The following identities are well-known; their proofs are elementary and can be found, e.g., in~\cite{Berezin-66} and \cite[Lemmas 7.8 and 7.12]{Solovej-notes}.

\begin{lemma}[Second quantizations of one- and two-body operators] \label{le:second-quantization} Let $\A$ be a symmetric operator on $\gH$ such that $u_n \in D(A)$ for all $n\ge 0$, and let $\omega$ be a symmetric operator on $\gH \otimes \gH$ such that $u_m \otimes u_n \in D(\omega)$ and $\langle u_m\otimes u_n, \omega \,u_p\otimes u_q \rangle = \langle u_n\otimes u_m, \omega \,u_p\otimes u_q \rangle$ for all $m,n,p,q \ge 0$. Then 
$$\dGamma (\A) := 0\oplus  \bigoplus_{N=1}^\infty \sum_{j=1}^N \A_j = \sum_{m,n\ge 0} \langle u_m,\A u_n\rangle_{\gH} \,\,a_m^* a_n$$
and 
$$ 0\oplus 0 \oplus \bigoplus_{N=2}^{\infty} \sum_{1\le i<j \le N} \omega_{ij} = \frac{1}{2} \sum_{m,n,p,q\ge 0} \langle u_m\otimes u_n, \omega \,u_p\otimes u_q \rangle_{\gH^2} \,\,a_m^* a_n^* a_p a_q$$  
as quadratic forms on the domain 
$$\bigcup_{M=0}^\infty \bigoplus_{N=0}^M \bigotimes_{\text{sym}}^N \left( {\rm Span} \{u_0,u_1,...\} \right) \subset \F.$$ 
\end{lemma} 

The same identities also hold for operators on the Fock space $\F_+$, where we can use the orthonormal basis $\{u_n\}_{n=1}^\infty$ for $\gH_+$. 
In particular, the {\it particle number operator} $\N= \dGamma(1)=\sum_{j=0}^\infty j 1_{\gH^j}$ on $\F$ can be rewritten as $\N= \sum_{n=0}^\infty a_n^*a_n$, and the {particle number operator} on $\F_+$ is  $\N_+=\sum_{n=1}^\infty a_n^*a_n.$ 

By using the second quantization, we can write $H_N:\gH^N\to\gH^N$ as
\bq \label{eq:second-quantization}
H_N = \left(\sum_{m,n\ge 0} T_{mn} a^*_m a_n  +\frac{1}{2(N-1)} \sum_{m,n,p,q\ge 0}W_{mnpq} a^*_m a^*_n a_p a_q\right)_{\big|\gH^N}.
\eq
Here $T_{mn}:= \left\langle { u_m, T u_n } \right\rangle$ and 
$$
W_{mnpq}:=\iint_{\Omega\times \Omega} {{\overline{u_m(x)}\overline{u_n(y)} {w(x-y)} u_p(x) u_q(y)}}\d x\d y.
$$

Since $H_N$ and $\bH$ live in different Hilbert spaces, to compare them we need to use the unitary transformation $U_N: \H^N\to \F^{\le N}_+$ defined before in (\ref{eq:def-unitary-UN}). The action of $U_N$ on annihilation and creation operators is given in the following lemma, whose proof is elementary and is left to the reader. 

\begin{proposition}[Properties of $U_N$] \label{le:act-UN} The operator $U_N$ defined in~\eqref{eq:def-unitary-UN} can be equivalently written as
\bq \label{eq:alt-def-UN}
U_N(\Psi)= \bigoplus_{j=0}^N Q^{\otimes j} \left(\frac{a_0^{N-j}}{\sqrt{(N-j)!}} \Psi\right)
\eq
for all $\Psi\in \gH^N$, and where $Q=1-|u_0\rangle\langle u_0|$. Similarly we have
\begin{equation}
U_N^* \left(\bigoplus_{j=0}^N \phi_j \right) = \sum _{j = 0}^N\frac{{{{(a_0^*)}^{N - j}}}}{{\sqrt {(N - j)!} }}{\phi _j}
\label{eq:alt-def-UN*}
\end{equation}
for all $\phi_j\in \gH_+^j$, $j=0,...,N$. These operators satisfy the following identities on $\F_+^{\le N}$:
\bqq
U_N \, a_0^* a_0 \,U_N^* &=& N- \N_+ ,\hfill\\
U_N\, a^*(f) a_0 \,U_N^* &=& a^*(f) \sqrt{N-\N_+},\hfill\\
U_N \,a_0^* a(f) \,U_N^* &=& \sqrt{N-\N_+}\, a(f),\hfill\\
U_N\, a^*(f) a(g) \,U_N^* &=& a^*(f) a(g),
\eqq
for all $f,g\in \gH_+$.
\end{proposition}

Note that the previous properties are purely algebraic. They do not depend on any special choice of the reference one-body function $u_0\in\gH$. Roughly speaking, the unitary transformation $U_N(\cdot) U_N^*$ leaves any $a_m^*$ or $a_m$ invariant when $m\geq1$, and it replaces each $a_0$, and $a_0^*$, by $\sqrt{N-\N_+}$. The latter is essentially the number $\sqrt{N}$ if $\N_+$ is small in comparison with $N$, on the considered sequence of states. 

By using these identities and the commutation relations (\ref{eq:CCR}), we can compute easily $U_N \A U_N^*$ where $\A$ is any operator on $\gH^N$ which is a (particle-conserving) polynomial in the creation and annihilation operators. For example, 
\bqq
U_N (a_0^* a_0^* a_m a_n)U_N^* &=& U_N (a_0^*a_m) U_N ^* \, U_N(a_0^*a_n)U_N^* \hfill\\
&=& \sqrt{N-\N_+}a_m \sqrt{N-\N_+} a_n.
\eqq
Using that $a(f)\sqrt{N-\N_+}=\sqrt{N-\N_+-1}\,a(f)$ for all $f\in \gH_+$, we obtain
$$U_N (a_0^* a_0^* a_m a_n)U_N^* = \sqrt{(N-\N_+)(N-\N_+-1)}\,a_m a_n.$$

A tedious but straightforward computation shows that 
\bq \label{eq:UHU-detail}
U_N H_N U_N^*- N e_{\rm H} = \sum_{j=0}^4 A_j 
\eq
where
\bqq
A_0 &=& \frac12 W_{0000}\frac{\N_+(\N_+-1)}{N-1}\hfill\\
A_1&=&\sum_{m\ge 1}\left( T_{0m} + W_{000m} \frac{N-\N_+ -1} {N-1} \right) \sqrt{N-\N_+} a_m \nn\hfill\\
&~&+ \sum_{m\ge 1}a_m^* \sqrt{N-\N_+}\left( T_{m0} + W_{m000} \frac{N-\N_+ -1} {N-1} \right) ,
\eqq
\bqq
A_2&=& \sum_{m,n\ge 1}  \left \langle u_m, \left( T-\mu_{\rm H}\right) u_n \right \rangle a_m^* a_n \nn\hfill\\
&~&+\sum_{m,n\ge 1}  \left \langle u_m, \left( |u_0|^2*w + K_1 \right) u_n \right \rangle a_m^* a_n \frac{N-\N_+} {N-1} \nn\hfill\\
&~&+ \frac{1}{2} \sum_{m,n\ge 1} \langle u_m \otimes u_n,K_2 \rangle a_n^* a_m^* \frac{\sqrt{(N-\N_+)(N-\N_+-1)}}{N-1}\nn\hfill\\
&~&+\frac{1}{2} \sum_{m,n\ge 1} \langle K_2, u_m\otimes u_n \rangle \frac{\sqrt{(N-\N_+)(N-\N_+-1)}}{N-1} a_m a_n ,
\eqq
\bqq
A_3 &=& \frac{1}{N-1} \sum_{m,n,p \ge 1} W_{mnp0} a_m^* a_n^* a_p \sqrt{N-\N_+} \nn\hfill\\
&~&+ \frac{1}{N-1} \sum_{m,n,p \ge 1} W_{0pnm} \sqrt{N-\N_+} a_p^* a_n a_m, \nn\hfill\\
A_4 &=& \frac{1}{2(N-1)} \sum_{m,n,p,q \ge 1} W_{mnpq} a_m ^* a_n ^* a_p a_q.
\eqq
Here recall that $e_{\rm H}=T_{00}+(1/2) W_{0000}$ and $\mu_{\rm H}=T_{00}+W_{0000}$. 

In the next section, we shall carefully estimate all the terms of the right side of (\ref{eq:UHU-detail}) to show that  
$$U_N H_N U_N -Ne_{\rm H} \approx \bH$$
in the regime $\N_+ \ll N$.

\section{Bound on truncated Fock space}\label{sec:bounds-truncated-Fock-space}

The main result in this section is the following bound. 

\begin{proposition}[Preliminary bound on truncated Fock space]\label{le:bound-truncated-Fock-space} Assume that {\rm (A1)} and {\rm (A2)} hold true. For any vector $\Phi$ in the quadratic form domain of $\bH$ such that $\Phi \in \F^{\le M}_+$ for some $1\le M\le N$, we have
\begin{align*}
\Big| \left\langle U_N (H_N-Ne_{\rm H})U_N^* \right\rangle _\Phi - \left\langle \mathbb{H} \right\rangle_\Phi \Big| \le C\sqrt {\frac{M}{N}} \left\langle  \mathbb{H} + C \right\rangle_\Phi .
\label{eq:bound-truncated-Fock-space} 
\end{align*}
\end{proposition}

Here for a self adjoint operator $A$ and an element $a$ in a Hilbert space, we write $\langle A \rangle _a$ instead of $\langle a, A a\rangle$ for short. Recall that the quadratic form domain of $\bH$ is the same as that of $\dGamma(h+1)$, on which the quadratic form $U_N (H_N-Ne_{\rm H})U_N^*$ is well-defined.
As an easy consequence of Proposition~\ref{le:bound-truncated-Fock-space}, we can prove the weak convergence in the first statement of Theorem~\ref{thm:lower-spectrum}.

\begin{corollary}[Weak convergence towards $\bH$]\label{cor:weak-CV-bH}
Assume that {\rm (A1)} and {\rm (A2)} hold true. Then we have for all fixed $\Phi,\Phi'$ in the quadratic form domain of the Bogoliubov Hamiltonian $\bH$,
\begin{equation}
\lim_{N\to\ii}\pscal{\Phi',U_N\big(H_N-N\,e_{\rm H}\big)U_N^*\,\Phi}_{\cF_+}=\pscal{\Phi',\bH\,\Phi}_{\cF_+}
\label{eq:weak-CVbis}
\end{equation}
where by convention $U_N^*$ is extended to $0$ outside of $\cF_+^{\leq N}$.
\end{corollary}

\begin{proof}[Proof of Corollary~\ref{cor:weak-CV-bH}]
It suffices to show the statement for $\Phi'=\Phi$. Note that the quadratic form domain of $\bH$ is the same as that of $\dGamma (h+1)$ because of (\ref{eq:comparison-bH-dGamma(h+1)}), and $\dGamma (h+1)$ preserves all the subspaces $\gH_+^m$'s. Therefore, if we denote by $\Phi_M$ the projection of $\Phi$ onto $\F_+^{\le M}$, then 
$$\lim_{M\to \infty} \langle \bH   \rangle_{\Phi_M} =\langle \bH  \rangle_\Phi\qquad\text{and}\qquad \lim_{M\to \infty} \lim_{N\to \infty} \langle \bH+C  \rangle _{\Phi_N-\Phi_M}=0.$$

If we denote $\widetilde{H}_N:=U_N(H_N-N\,e_{\rm H})U_N^*$, then from Proposition \ref{le:bound-truncated-Fock-space} we have 
$$\lim_{M\to \infty} \langle {\widetilde {H}_N - \bH} \rangle_{\Phi_M }= 0 \quad {\rm and} \lim_{M\to \infty} \lim_{N\to \infty} \quad \langle {\widetilde {H}_N} \rangle_{\Phi_N -\Phi_M }= 0.$$
The latter convergence still holds true with $\widetilde {H}_N$ replaced by a non-negative operator $H_N':=\widetilde {H}_N+ C_0 (\bH+C_0)$, where $C_0>0$ is chosen large enough. By using the Cauchy-Schwarz inequality for the operator $H_N'$, we deduce that
$$\lim_{M\to \infty} \lim_{N\to \infty} \langle {\widetilde {H}_N} \rangle_{\Phi}-\langle {\widetilde {H}_N} \rangle_{\Phi_M} =0.$$
We then can conclude that $ \langle {\widetilde {H}_N} \rangle_{\Phi} \to \langle {\bH} \rangle_{\Phi}$ as $N\to \infty$.
\end{proof}

The rest of the section is devoted to the proof of Proposition~\ref{le:bound-truncated-Fock-space}. We shall need the following technical result.

\begin{lemma}\label{le:simple-bound-excited-particle} If {\rm (A1)-(A2)} hold true, then we have the operator inequalities on $\F_+$:
\bqq 
\dGamma (QTQ) &\le &  \frac{1}{1-\alpha_1} \mathbb{H}+ C\N_+ + C, \hfill\\
\dGamma ( Q(|u_0|^2*|w|)Q) &\le & \frac{\alpha_2}{1-\alpha_1} \mathbb{H}+C\N_+ +  C,
\eqq
where $1>\alpha_1>0$ and $\alpha_2>0$ are given in the relative bound (\ref{eq:bound-w-1}) in Assumption {\rm (A1)}.
\end{lemma}

\begin{proof} Using (\ref{eq:bound-w-1}) we get $|u_0|^2*w \ge -\alpha_1 (T+C)$ and $|u_0|^2*|w| \le \alpha_2 (T+C)$. Consequently, $T \le (1-\alpha_1)^{-1} h+C$ and $|u_0|^2*|w| \le  \alpha_2(1-\alpha_1)^{-1}h +C$. The desired estimates  follows from the lower bound $ \mathbb{H}\ge \dGamma(h)-C\N_+-C$ (see Remark \ref{rmk:bH-CdGamma(h)} in Appendix \ref{apd:Bogoliubov-Hamiltonian}).
\end{proof}

Now we give the

\begin{proof}[Proof of Proposition \ref{le:bound-truncated-Fock-space}] Let $\Phi$ be a normalized vector in the quadratic form domain of $\bH$ such that $\Phi \in \F^{\le M}_+$ for some $1\le M\le N$. Starting from the identity (\ref{eq:UHU-detail}), we shall compare $\langle A_2\rangle_\Phi$ with $\langle \bH\rangle_\Phi$, and show that the
remaining part is negligible. 

\subsubsection*{Step 1. Main part of the Hamiltonian}\label{eq:main-part} 
\begin{lemma}[Bound on $A_2-\bH$]\label{le:bound-A2} We have
\bqq
\Big| \langle A_2 \rangle_\Phi -\langle \bH \rangle_\Phi \Big| \le \frac{M}{N-1} \left( \frac{\alpha_2}{1-\alpha_1} \langle \mathbb{H} \rangle_\Phi+C\langle \N_+ +1 \rangle_\Phi   \right).
\eqq
\end{lemma}

\begin{proof}

From (\ref{eq:Bogoliubov-Hamiltonian-m-n}) and (\ref{eq:UHU-detail}) we have
\bq \label{eq:A2-bH}
\langle A_2 \rangle_\Phi -\langle \bH \rangle_\Phi  &=&  - \left\langle \dGamma\left(Q(|u_0|^2*w)Q+K_1 \right) \frac{\N_+-1}{N-1} \right\rangle_\Phi \nn\hfill\\
&~& + \Re  \sum_{m,n\ge 1} \langle u_m \otimes u_n, K_2 \rangle  \left \langle  a_m^*a_n^* X\right \rangle_\Phi
\eq
where $X:= \sqrt{(N-\N_+)(N-\N_+-1)}/(N-1)-1$. 

Since $\dGamma\left(Q(|u_0|^2*w)Q+K_1 \right)$ commutes with $\N_+$, we have
$$
0\le \dGamma\left(Q(|u_0|^2*w)Q+K_1 \right)\frac{\N_+-1}{N-1} \le \frac{M}{N-1}\dGamma\left(Q(|u_0|^2*w)Q+K_1 \right)
$$
on $\F_+^{\le M}$. By using Lemma \ref{le:simple-bound-excited-particle} and the boundedness of $K_1$, we can bound the first term of the right side of (\ref{eq:A2-bH}) as
\bq \label{eq:A2-bH-first-term}
\left| \left\langle \dGamma\left(Q(|u_0|^2*w)Q+K_1 \right)\frac{\N_+-1}{N-1}\right\rangle_\Phi \right| \quad \quad \quad \quad \quad \quad \nn\\
\le \frac{M}{N-1} \left( \frac{\alpha_2}{1-\alpha_1} \langle \mathbb{H} \rangle_\Phi +C\langle \N_+ \rangle_\Phi +  C \right). 
\eq

The second term of the right side of (\ref{eq:A2-bH}) can be estimated using the Cauchy-Schwarz inequality
\bq \label{eq:A2-bH-second-term}
&~&  \Big|  \sum_{m,n\ge 1} \langle u_m \otimes u_n, K_2 \rangle \left \langle a_m^* a_n^* X \right \rangle_\Phi \Big|\hfill \nn \\
 &\le & \left(  \sum_{m,n\ge 1}\!\! |\langle u_m \otimes u_n, K_2 \rangle|^2 \right)^{1/2} \left(  \sum_{m,n\ge 1} \left\langle a_m^* a_n^* a_n a_m \right\rangle _\Phi    \right)^{1/2} \left\langle X^2 \right\rangle_\Phi^{1/2} \nn\hfill \\ 
&\le & \left( \int_{\Omega}\int_{\Omega} |K_2(x,y)|^2 \d x \d y \right)^{1/2}  \left\langle \N_+^2 \right\rangle _\Phi^{1/2}  \left\langle \frac{4(\N_++1)^2}{(N-1)^2} \right\rangle_\Phi^{1/2} \nn\hfill\\
&\le & \frac{CM}{N-1} \left\langle \N_+ +1 \right\rangle_\Phi
\eq
Here we have used the fact that the operator $K_2$ is Hilbert-Schmidt and the inequality $(X-1)^2\le 4(\N_++1)^2/(N-1)^2$ and the estimate $\N_+\le M$ on $\F_+^{\le M}$. From (\ref{eq:A2-bH}), (\ref{eq:A2-bH-first-term}) and (\ref{eq:A2-bH-second-term}), the bound in Lemma \ref{le:bound-A2} follows. 
\end{proof}
\subsubsection*{Step 2. Unimportant parts of the Hamiltonian} Now we estimate the other terms of the right side of (\ref{eq:UHU-detail}). First at all, by using $ \N_+^2 \le M \N_+$ on $\F_+^{\le M}$ we get 
\bq \label{eq:bound-A0}
\left| \langle A_0\rangle_\Phi \right|=\left| W_{0000} \frac{ \langle{\N_+(\N_+-1)}\rangle _{\Phi} }{N-1} \right| \le \frac{CM}{N} \langle \N_+ \rangle_\Phi.
\eq
The terms $\langle A_1\rangle_\Phi$, $\langle A_4 \rangle_\Phi$ and $\langle A_3 \rangle_\Phi$ are treated in the next lemmas.

\begin{lemma}[Bound on $A_1$] \label{le:bound-A1} We have
$$
   \big| \langle A_1 \rangle_\Phi \big| \le   C\sqrt{\frac{M}{N}} \langle \N_+ \rangle_\Phi  .
   $$
\end{lemma}

\begin{proof} By using Hartree's equation 
$\left( T+|u_0|^2*w -\mu_{\rm H} \right) u_0=0$
we obtain
\[ T_{m0}=\left\langle {{u_m},T{u_0}} \right\rangle  =  - \left\langle {{u_m},(|{u_0}{|^2}*w){u_0}} \right\rangle  =  - {W_{m000}}
\]
for all $m\ge 1$. Therefore,
\[ 
\langle A_1 \rangle_\Phi =  -2 \Re   \sum_{m\ge 1} \frac{W_{m000}}{N-1} \left \langle a_m^* \sqrt{N-\N_+} \N_+ \right \rangle_\Phi.
\]
By using the Cauchy–Schwarz inequality we get
\[\begin{gathered}
   \big| \langle A_1 \rangle_\Phi \big| \le   \frac{2}{N-1} \left(\! \sum_{m \ge 1} \!|W_{m000}|^2 \right)^{1/2} \!\!\!\!\left( \! \sum_{m\ge 1}\left \langle a_m^* a_m \right \rangle_\Phi \! \left \langle  \N_+^2(N-\N_+)\right \rangle_\Phi \!\! \right)^{1/2} \hfill \\
   \quad \quad \quad  \le   C\sqrt{\frac{M}{N}} \langle \N_+ \rangle_\Phi  .\hfill \\ 
\end{gathered} \]
Here in the last estimate we have used the bound
\bqq
\sum_{m\ge 1} |W_{m000}|^2 &=& \sum_{m\ge 1} \left| \int_{\Omega}\int_{\Omega} u_m(x) u_0(y) u_0(x)u_0(y) w(x-y) \d x \d y \right|^2 \\
&\le &   \int_{\Omega}\int_{\Omega} |u_0(x)|^2 |u_0(y)|^2 w(x-y)^2 \d x \d y  <\infty
\eqq
and the inequality $\N_+^2(N-\N_+) \le MN \N_+$ on $\F_+^{\le M}$.
\end{proof}

\begin{lemma}[Bound on $A_4$]\label{le:bound-A4} We have 
\bqq
\big| \langle A_4 \rangle_\Phi \big| \le  \frac{M}{N-1} \left( \frac{\alpha_2}{1-\alpha_1} \langle \bH \rangle_\Phi  + C\langle \N_+ \rangle_\Phi +C \right) .
\eqq
\end{lemma}

\begin{proof} From Assumption (A1) we have
$$ Q\otimes Q \left(\alpha_2 (T_x + T_y +C) -w(x-y) \right) Q\otimes Q \ge 0$$
on $\gH^2$, where $Q=1-| u_0 \rangle \langle u_0 |$. By taking the second quantization (see Lemma \ref{le:second-quantization}) of this two-body operator, we can bound $A_4$ from above by 
\[\begin{gathered}
   \frac{1}{2(N-1)} \sum_{m,n,p,q\ge 1} \left \langle u_m\otimes u_n, \alpha_2 ( T\otimes 1 +  1\otimes T +C) u_p\otimes u_q  \right \rangle a_m^* a_n^* a_p a_q\hfill \\
  = \frac{\alpha_2}{N-1}\dGamma (QTQ) (\N_+-1) + \frac{C}{2(N-1)} \N_+(\N_+-1).\hfill \\ 
\end{gathered} \]
By using Lemma \ref{le:simple-bound-excited-particle} and the fact that $\dGamma (QTQ)$ commutes with $\N_+$, we obtain the operator inequality
$$ A_4 \le \frac{M}{N-1} \left( \frac{\alpha_2} {1-\alpha_1} \bH + C\N_+ +C \right)~~{\rm on}~\F_+^{\le M}.$$
Employing this argument with $w$ replaced by $-w$, we obtain the same upper bound with $A_4$ replaced by $-A_4$, and the bound in Lemma \ref{le:bound-A4} follows.
\end{proof}

\begin{lemma}[Bound on $A_3$] \label{le:bound-A3} We have
\[ \left| \langle A_3 \rangle_\Phi \right| \leqslant \sqrt{\frac{M}{N-1}}\left( \frac{\alpha_2}{1-\alpha_1} \langle \bH \rangle_\Phi  + C\langle \N_+ \rangle_\Phi +C \right) .\]
\end{lemma}

\begin{proof} Let us write $w=w_+ - w_-$ where $w_+=\max\{w,0\}$ and $w_-=\max\{-w,0\}$. By taking the second quantization (see Lemma \ref{le:second-quantization}) of the non-negative two-body operator 
\bqq
 \Big( Q\otimes (Q-\eps P) w_+ Q\otimes (Q-\eps P) + (Q-\eps P) \otimes Q w_+ (Q-\eps P) \otimes Q \Big) \quad \quad \hfill\\
 \quad \quad + \Big(  Q\otimes (Q+\eps P) w_- Q\otimes (Q+\eps P) + (Q+\eps P) \otimes Q w_- (Q+\eps P) \otimes Q \Big)
\eqq
where $Q=1-P=1-| u_0 \rangle \langle u_0|$ and $\eps>0$, we obtain the following Cauchy-Schwarz  inequality
\bqq
 &~&  \sum_{m,n,p \ge 1}  \langle u_m \otimes u_n, w u_p \otimes u_0 \rangle a_m^* a_n^* a_p a_0 \quad \quad\quad\quad\quad\quad \hfill\\
&~&\quad\quad\quad\quad\quad\quad+ \sum_{m,n,p \ge 1} \langle u_0 \otimes u_p, w u_n \otimes u_m \rangle a_0^* a_p^* a_n a_m \hfill \\
 &\le &  \eps^{-1}\sum_{m,n,p,q \ge 1} \langle u_m \otimes u_n, |w| u_p \otimes u_q \rangle a_m^* a_n^* a_p a_q \quad \hfill\\
&~&  \quad\quad\quad\quad\quad\quad+ \eps \sum_{m,n \ge 1} \langle u_m \otimes u_0, |w| u_n \otimes u_0 \rangle a_m^* a_0^* a_n a_0.
\eqq
After performing the unitary transformation $U_N (\cdot) U_N^*$ we find that
\bq \label{eq:upper-bound-A3}
A_3 &\le &  \frac{1}{\eps(N-1)}\sum_{m,n,p,q \ge 1} \langle u_m \otimes u_n, |w| u_p \otimes u_q \rangle a_m^* a_n^* a_p a_q \quad \nn\hfill\\
&~&  \quad\quad\quad\quad+ \frac{\eps}{N-1} \dGamma (Q(|u_0|^2*|w|)Q) (N-\N_+).
\eq
for every $\eps>0$. 

On the other hand, by the same proof of Lemma \ref{le:bound-A4} we have
\bq \label{eq:A4-abs}
\frac{1}{N-1}\sum_{m,n,p,q \ge 1} \langle u_m \otimes u_n, |w| u_p \otimes u_q \rangle \langle a_m^* a_n^* a_p a_q \rangle_\Phi \\
\le \frac{M}{N-1} \left( \frac{\alpha_2}{1-\alpha_1} \langle \bH \rangle_\Phi  + C\langle \N_+ \rangle_\Phi +C \right).
\eq
Moreover, by using Lemma \ref{le:simple-bound-excited-particle} we get 
$$\dGamma(Q(|u_0|^2*|w|)Q)(N-\N_+) \le (N-1) \dGamma(Q(|u_0|^2*|w|)Q)$$
and 
\bq \label{eq:bound-direct}
 \left \langle \dGamma (Q(|u_0|^2*|w|)Q) \frac{N-\N_+}{N-1} \right \rangle _\Phi &\le & \left \langle \dGamma (Q(|u_0|^2*|w|)Q)  \right \rangle _\Phi 
\hfill\\
&\le &   \left( \frac{\alpha_2}{1-\alpha_1} \langle \bH \rangle_\Phi  + C\langle \N_+ \rangle_\Phi +C \right).\nn
\eq
From (\ref{eq:upper-bound-A3}), (\ref{eq:A4-abs}) and (\ref{eq:bound-direct}), we can deduce that
\bqq 
\langle A_3 \rangle_\Phi \le \sqrt{\frac{M}{N-1}}\left( \frac{\alpha_2}{1-\alpha_1} \langle \bH \rangle_\Phi  + C\langle \N_+ \rangle_\Phi +C \right).
\eqq
By repeating the above proof with $w$ replaced by $-w$, we obtain the same upper bound on $-\langle A_3 \rangle_\Phi$ and then finish the proof of Lemma \ref{le:bound-A3}. 
\end{proof}

\subsubsection*{Step 3. Conclusion}

Using Lemmas \ref{le:bound-A2}, \ref{le:bound-A1}, \ref{le:bound-A4}, \ref{le:bound-A3} and (\ref{eq:bound-A0}), we can conclude from (\ref{eq:UHU-detail}) that
$$
   \Big| {{{\left\langle {{U_N}(H_N-Ne_{\rm H})U_N^*} \right\rangle }_\Phi } - {{\left\langle \mathbb{H} \right\rangle }_\Phi }} \Big| \le \sqrt{\frac{M}{N-1}}\left(\frac{3\alpha_2}{1-\alpha_1}\langle \bH \rangle_\Phi  + C\langle \N_+ +1 \rangle_\Phi \right)
   $$
where $\alpha_1$ and $\alpha_2$ are the constants appearing in the relative bound in (A1). Since $\N_+\le C(\bH+C)$ due to (\ref{eq:comparison-bH-dGamma(h+1)}), the estimate in Proposition \ref{le:bound-truncated-Fock-space} follows.
\end{proof}

\section{Localization in Fock space}\label{sec:localization}

From Proposition \ref{le:bound-truncated-Fock-space}, we have
$$ \lim_{N\to \infty} \langle  U_N H_N U_N^* -Ne_H - \mathbb{H}\rangle_\Phi =0 $$
for every $\Phi\in \F_+^{\le M}$ with $M$ fixed. In the next step, we want to localize a state in $\F_+^{\le N}$ into the smaller truncated Fock space $\F_+^{\le M}$ with $M \ll N$, without changing the energy too much. The following result is an adaption of the localization method used by Lieb and Solovej in \cite[Theorem A.1]{LieSol-01}. 

\begin{proposition}[Localization of band operators in $\F_+$]\label{pro:localization} Let $\A$ be a non-negative operator on $\F$ such that $P_jD(\A)\subset D(\A)$ and $P_i \A P_j=0$ when $|i-j|$ is larger than some constant $\sigma$, where $P_j$ is the projection onto $\gH_+^j$. Let $0\leq f,g\leq 1$ be smooth, real functions such that $f^2+g^2\equiv1$, $f(x)=1$ for $|x|\leq 1/2$ and $f(x)=0$ for $|x|\geq 1$ and let $f_M$ and $g_M$ be the localization operators on $\F_+$ defined by
$$ f_M =f (\N_+/M)=\sum_{j=0}^\infty f(j/M) P_j \quad {and}\quad g_M=g(\N_+/M)=\sum_{j=0}^\infty g(j/M) P_j.$$

\noindent $(i)$ We have
\bq \label{eq:IMS-band-operator}
 - \frac{C_f \sigma^3}{M^2} \A_0\le \A-f_M \A f_M -g_M \A g_M \le \frac{C_f \sigma^3}{M^2} \A_0
 \eq
where $\A_0=\sum P_j \A P_j$ is the diagonal part of $\A$ and $C_f=2(\|f'\|^2_{\infty}+\|g'\|^2_{\infty})$.
\smallskip

\noindent $(ii)$ Let $Y$ be a finite-dimensional subspace of $D(\A)$ such that $||g_M\Phi ||^2<(\dim Y)^{-1}$ for every normalized vector $\Phi\in Y$, then $\dim (f_M Y)=\dim Y$.
\end{proposition}

The reader should really think of the discrete Laplacian ($\sigma=1$) in which case the inequality (\ref{eq:IMS-band-operator}) is nothing but a discrete version of the IMS formula \cite[Theorem 3.2]{CycFroKirSim-87}. The statement on the dimension of the localized space $f_M Y$ will be useful to control eigenvalues via the min-max formula. A proof of Proposition \ref{pro:localization} can be found in Appendix \ref{apd:localization}. In the following we have two applications of this result in our particular situation.

Our first application is a localization for $\mathbb{H}$.

\begin{lemma}[Localization for $\mathbb{H}$] \label{le:localization-H} For every $1\le M \le N$ we have
$$-\frac{C}{M^2} (\bH +C ) \le \bH - f_M \bH f_M - g_M \bH g_M  \le \frac{C}{M^2} (\bH +C ). $$
\end{lemma}

\begin{proof} We can apply Proposition \ref{pro:localization} with $\A=\mathbb{H}-\lambda_1(\mathbb{H})$ and $\sigma=2$. Note that the diagonal part of $\bH$ is nothing but $\dGamma(h+K_1)$ which can be bounded from above by $C(\bH+C)$, due to  (\ref{eq:comparison-bH-dGamma(h+1)}).
\end{proof}

Now we turn to a localization for $\widetilde {H}_N := U_N (H_N-Ne_{\rm H}) U_N^*$.

\begin{lemma}[Localization for $\widetilde {H}_N $]\label{le:localization-HN} For every $1\le M \le N$ we have 
\bqq -\frac{C}{M^2} \Big( C\dGamma(h)_{| \F_+^{\le N}}+C-\lambda_1(\widetilde {H}_N) \Big) \le  \widetilde {H}_N -  f_M \widetilde {H}_N f_M - g_M \widetilde {H}_N g_M \le \hfill\\ \le   \frac{C}{M^2} \Big(C\dGamma(h)_{| \F_+^{\le N}}+C -\lambda_1(\widetilde {H}_N) \Big).
\eqq
\end{lemma}

\begin{proof} We apply Proposition \ref{pro:localization} with $\A=\widetilde {H}_N -\lambda_1(\widetilde {H}_N)$ and $\sigma=2$. From the inequality $\widetilde {H}_N \le C(\bH+C)$ due to Proposition \ref{le:bound-truncated-Fock-space}, we see that the diagonal part of $\widetilde {H}_N$ can be bounded from above by $C\dGamma(h)_{| \F_+^{\le N}}+C$. 
\end{proof}
\begin{remark} \label{rmk:localization-HN} \it The error term $\dGamma(h)_{| \F_+^{\le N}}$ in Lemma \ref{le:localization-HN} can be further replaced by $C(\bH+C)$ due to (\ref{eq:comparison-bH-dGamma(h+1)}), or replaced by $C(\widetilde{H}_N +CN)$ due to the stability 
$$H_N \ge (1-\alpha_1) \sum_{i=1}^N T_i -C N \ge \frac{1-\alpha_1}{1+\alpha_2} \sum_{i=1}^N h_i -CN\quad{\rm on}~\gH^N,$$
where $\alpha_1$ and $\alpha_2$ are given in Assumption (A1). The bound $d\Gamma(h)\le C(\widetilde{H}_N+CN)$ on $\F_+^{\le N}$ also implies that $\lambda_1(\widetilde{H}_N) \ge -CN$, although it is not optimal (we shall see that $\lambda_1(\widetilde{H}_N)$ is of order $O(1)$). 
\end{remark}

\section{Proof of Main Theorems} \label{sec:proof-Main-Theorems}

We recall that we have introduced the notation $\widetilde {H}_N := U_N (H_N-Ne_{\rm H}) U_N^*$.

\subsection{Proof of Theorem \ref{thm:lower-spectrum}} \label{sec:proof-Theorem-lower-spectrum}

The first statement of Theorem~\ref{thm:lower-spectrum} was already proved in Corollary~\ref{cor:weak-CV-bH} above. We now turn to the proof of the other statements, using the localization method.

\subsubsection*{Step 1. Convergence of min-max values.}

\paragraph*{\it Upper bound.} By applying Lemma \ref{le:localization-H} and Proposition \ref{le:bound-truncated-Fock-space}, we have, for every normalized vector $\Phi\in D(\bH)$ and for every $1\le M \le N$, 
\bq \label{eq:upper-bound-eq1}
\bH &\ge & f_M \mathbb{H} f_M + g_M \bH g_M -\frac{C}{M^2} (\bH +C ) \nn \hfill\\
&\ge & f_M \left[ \left( 1+ C\sqrt{\frac{M}{N}}\right)^{-1} \widetilde{H}_N -C\sqrt{\frac{M}{N}} \right] f_M \nn\hfill\\
&~&+ \lambda_1(\bH) g_M^2  - \frac{C}{M^2} (\bH+C ) .
\eq

Now fix an arbitrary $L\in \mathbb{N}$. For every $\delta>0$ small, we can find an $L$-dimensional subspace $Y \subset D(\bH)\subset \F_+$ such that 
$$\max_{\Phi\in Y, ||\Phi||=1}\left\langle {\mathbb{H}} \right\rangle_{\Phi} \le  \lambda_L(\mathbb{H})+\delta.$$
By using the inequality $g_M^2\le 2\N_+/M$ and $\N_+ \le C(\bH+C)$, which is due to (\ref{eq:comparison-bH-dGamma(h+1)}), we get $||g_M \Phi||^2 \le  C_L/M$ for every normalized vector $\Phi\in Y$. Here $C_L$ is a constant depending only on $\lambda_L(\bH)-\lambda_1(\bH)$. Therefore, if $M > LC_L$, then $\dim (f_M Y)=L$ by Proposition \ref{pro:localization} (ii). Consequently, by the min-max principle,
$$ \max_{\Phi\in Y, ||\Phi||=1} \frac{\langle f_M \Phi, \widetilde{H}_N f_M \Phi\rangle}{||f_M \Phi||^2} \ge \lambda_L(\widetilde{H}_N).$$
Thus from (\ref{eq:upper-bound-eq1}) we obtain, for $M$ large enough,
$$ \lambda_L(\bH) \ge \left( 1-\frac{C_L}{M} \right) \left[ \left( 1+ C\sqrt{\frac{M}{N}}\right)^{-1} \lambda_L (\widetilde{H}_N) -C\sqrt{\frac{M}{N}} \right] - \frac{C_L}{M}.$$
By choosing $M=N^{1/3}$ we get the upper bound
\bq \label{eq:min-max-upper-bound}
\lambda_L (H_N)-Ne_{\rm H}=\lambda_L (\widetilde {H}_N) \le \lambda_L(\mathbb{H})+ C_L N^{-1/3}.
\eq

\paragraph*{\it Lower bound.} By applying Lemma \ref{le:localization-HN} and Proposition \ref{le:bound-truncated-Fock-space}, for every $1\le M \le N$ we have the operator inequalities on $\F_+^{\le N}:$ 
\bq \label{eq:lower-bound-eq1}
\widetilde{H}_N &\ge & f_M \widetilde{H}_N f_M + g_M \widetilde{H}_N g_M - \frac{C}{M^2} \Big(C\dGamma(h)_{| \F_+^{\le N}} + C-\lambda_1(\widetilde{H}_N) \Big) \nn \hfill\\
&\ge & f_M \left[ \left( 1- C\sqrt{\frac{M}{N}}\right) \bH -C\sqrt{\frac{M}{N}} \right] f_M \nn\hfill\\
&~&+ \lambda_1(\widetilde{H}_N) g_M^2  - \frac{C}{M^2} \Big(C\dGamma(h)_{| \F_+^{\le N}} + C-\lambda_1(\widetilde{H}_N) \Big) 
\eq

We first show that $\lambda_1(\widetilde{H}_N) \ge \lambda_1(\bH)+o(1)$. For every $\delta>0$ small, we can find a normalized vector $\Phi_N \in D(\widetilde {H}_N)\subset \F_+^{\le N}$ such that 
$$\langle { \widetilde {H}_N } \rangle_{\Phi_N} \le  \lambda_1(\widetilde {H}_N)+\delta.$$
From Assumption (A3) we have
$$ 
||g_M \Phi_N||^2 \le  \frac{2\left\langle {\N_+} \right\rangle_{\Phi_N}}{M} \le \frac{2N\eps_R(N)}{M}
$$
where $\eps_R$ is as in (A3). We can choose $R=1$. In particular, if we choose $M \gg N \eps_R(N)$, then $||g_M \Phi_N||^2 \to 0$ as $N\to \infty$ independently of $\delta$ and the choice of $\Phi_N$.  Thus from  (\ref{eq:lower-bound-eq1}) and the simple bound $\dGamma(h)_{| \F_+^{\le N}} -\lambda_1(\widetilde{H}_N)\le C(\widetilde{H}_N +CN)$ (see Remark \ref{rmk:localization-HN} after Lemma \ref{le:localization-HN}) we can conclude that
\bqq 
\lambda_1( \widetilde{H}_N ) +\delta &\ge & ||f_M \Phi_N||^2 \left[ \left( 1- C\sqrt{\frac{M}{N}}\right) \lambda_1(\bH) -C\sqrt{\frac{M}{N}} \right]  \nn\hfill\\
&~&+ \lambda_1(\widetilde{H}_N) ||g_M \Phi_N ||^2  - \frac{CN}{M^2} .
\eqq
By choosing $M$ such that $ \max\{N \eps_R(N),\sqrt{N}\} \ll M \ll N$, then taking $\delta\to 0$, we obtain the lower bound $\lambda_1(\widetilde{H}_N) \ge \lambda_1(\bH)+o(1)$. 

By adapting the above argument, we can show that $\lambda_L(\widetilde{H}_N) \ge \lambda_L(\bH)+o(1)$ for an arbitrary $L\in \mathbb{N}$. In fact, for every $\delta>0$ small, we can find an $L$-dimensional subspace $Y\subset D(\widetilde {H}_N)\subset \F_+^{\le N}$ such that 
$$\max_{\Phi\in Y, ||\Phi||=1}\langle { \widetilde {H}_N } \rangle_{\Phi} \le  \lambda_L(\widetilde {H}_N)+\delta.$$
From Assumption (A3) and the upper bound $\lambda_L(\widetilde{H}_N)-\lambda_1(\widetilde{H}_N)\le \lambda_L(\bH)-\lambda_1(\bH)+o(1)$, we have
$$ 
||g_M \Phi||^2 \le  \frac{2\left\langle {\N_+} \right\rangle_{\Phi}}{M} \le \frac{N\eps_R(N)}{M}
$$
for every normalized vector $\Phi\in Y$, where $\eps_R$ is as in (A3). We can choose $R=\lambda_L(\bH)-\lambda_1(\bH)+1$. In particular, if $M \gg N \eps_R(N)$, then $||g_M \Phi||^2 \to 0$ as $N\to \infty$ independently of $\eps$ and the choice of $\Phi$ in $Y$. Consequently, when $N$ is large enough we have $\dim (f_M Y)=L$ by Proposition \ref{pro:localization} (ii), and by the min-max principle, we get
$$ \max_{\Phi\in Y, ||\Phi||=1} \frac{\langle f_M \Phi, \bH f_M \Phi\rangle}{||f_M \Phi||^2} \ge \lambda_L(\bH).$$
Thus from  (\ref{eq:lower-bound-eq1}) and the simple bound $\dGamma(h)_{| \F_+^{\le N}} -\lambda_1(\widetilde{H}_N)\le C(\widetilde{H}_N +CN)$ (see Remark \ref{rmk:localization-HN} after Lemma \ref{le:localization-HN}) we get
\bq \label{eq:lower-bound-bad-error}
\lambda_L(\widetilde{H}_N) \ge  \left( 1- C\sqrt{\frac{M}{N}}\right) \lambda_L(\bH) -C\sqrt{\frac{M}{N}} -C_L \frac{N \eps_R(N)}{M} - C_L \frac{N}{M^2}.
\eq
By choosing $M$ such that $ \max\{N \eps_R(N),\sqrt{N}\} \ll M \ll N$, we obtain  $\lambda_L(\widetilde{H}_N) \ge \lambda_L(\bH)+o(1)$. 

\medskip

\begin{remark}[Remark on the convergence rate] \it The error obtained in the lower bound (\ref{eq:lower-bound-bad-error}) is not better than $\sqrt[3]{\eps_R(N)}+N^{-1/5}$. However, it can be improved by the following bootstrap argument. First, from (\ref{eq:lower-bound-eq1}) with $M=rN$ for some small fixed number $r>0$, and the simple bound $\dGamma(h)_{| \F_+^{\le N}} -\lambda_1(\widetilde{H}_N)\le C(\widetilde{H}_N +CN)$ (see Remark \ref{rmk:localization-HN} after Lemma \ref{le:localization-HN}), we obtain
\bq \label{eq:improved-bound-fH}
f_{rN} \bH f_{rN} \le C(\widetilde{H}_N + C).
\eq
Next, by projecting the inequality (\ref{eq:lower-bound-eq1}), with $1\ll M \ll N$, onto the subspace $f_{rN}\F_+$ and using the refined bound $\dGamma(h)_{| \F_+^{\le N}} -\lambda_1(\widetilde{H}_N)\le C(\bH +C)$ (see Remark \ref{rmk:localization-HN}), we get
\bq  \label{eq:second-localization-fHNf}
f_{rN}\widetilde{H}_N f_{rN} &\ge & f_M \left[ \left( 1- C\sqrt{\frac{M}{N}}\right) \bH -C\sqrt{\frac{M}{N}} \right] f_M  \nn\hfill\\
&~&+ \lambda_1(\widetilde{H}_N) g_M^2 f_{rN}^2  - \frac{C}{M^2} f_{rN}(\bH + C) f_{rN}
\eq
Here we have used the fact that $f_M f_{rN}=f_M$ when $M\le rN/2$. Finally, we can use Lemma \ref{le:localization-HN} with $M=rN$, then estimate $f_{rN}\widetilde{H}_N f_{rN}$ by (\ref{eq:second-localization-fHNf}), and employ the inequality $g_M^2 \le 2\N_+/M \le C(\bH+C)/M$ and (\ref{eq:improved-bound-fH}). We have
 \bqq
\widetilde{H}_N &\ge & f_M \left[ \left( 1- C\sqrt{\frac{M}{N}}\right) \bH -C\sqrt{\frac{M}{N}} \right] f_M  - C\frac{\N_+}{N} - \frac{C}{M}(\widetilde{H}_N+C)
\eqq
when $1\ll M \ll N$. Consequently,   
$$
\lambda_L(\widetilde{H}_N) \ge \lambda_L(\bH) - C_L\left( \sqrt{\frac{M}{N}}+\frac{1}{M} + \eps_R(N) \right)
$$
In the latter bound, by choosing $M=N^{1/3}$, we get 
$$\lambda_L(\widetilde{H}_N) \ge \lambda_L(\bH) -C(\eps_R(N)+N^{-1/3}).$$      
\end{remark}

\paragraph*{\it Spectral gap.} By (i) one has for every $L=1,2,...$, 
$$ \liminf_{N\to \infty} \left( \lambda_{L}(H_N)-\lambda_1(H_N) \right) = \lambda_L(\mathbb{H})-\lambda_1(\mathbb{H})$$
Taking the limit as $L\to \infty$ we obtain the spectral gap
$$
\liminf_{N\to \infty} \left( \inf\sigma_{\rm ess} (H_N) - \lambda_1(H_N) \right) \ge \inf\sigma_{\rm ess}(\mathbb{H})-\lambda_1(\mathbb{H}).
$$

\subsubsection*{Step 2. Convergence of lower eigenvectors}

First we consider the convergence of ground states. If $\Psi_N$ is a ground state of $H_N$, then $\Phi_N:= U_N \Psi_N$ is a ground state of $\widetilde {H}_N := U_N (H_N-Ne_{\rm H}) U_N^*$ on $\F_+^{\le N}$. From the proof in Step 1, if $\max\{N\eps_R(N), \sqrt{N}\}\ll M \ll N$, then we have $||g_M \Phi_N||^2 \to 0$ and
$$\lim_{N\to \infty}\left\langle f_M \Phi_N , \mathbb{H} f_M \Phi_N \right\rangle=\lambda_1(\mathbb{H})= \left\langle \Phi^{(1)} ,\mathbb{H} \Phi^{(1)}\right\rangle$$
where $\Phi^{(1)}$ is the unique ground state of $\mathbb{H}$ on $\F_+$.

To prove $\Phi_N\to \Phi^{(1)}$, it suffices to show that $f_M \Phi_N\to \Phi^{(1)}$. Let us write 
$$f_M \Phi_N=a_N +b_N$$
where $a_N\in {\rm Span}(\{\Phi^{(1)}\})$ and $b_N \bot \Phi^{(1)}$. Then
\bqq
 \langle f_M\Phi_N, \mathbb{H} f_M\Phi_N \rangle &=& \langle a_N, \mathbb{H} a_N \rangle + \langle b_N, \mathbb{H} b_N \rangle \ge \lambda_1(\mathbb{H})||a_N||^2 + \lambda_2(\mathbb{H})||b_N||^2 \hfill\\
 &=& ||f_M \Phi_N||^2 \lambda_1(\bH )+ (\lambda_2 ( \bH )-\lambda_1 (\bH ) )||b_N||^2.
 \eqq
Since $\left\langle f_M \Phi_N, \mathbb{H} f_M \Phi_N \right\rangle_{\widetilde\Phi_N}- ||f_M\Phi_N||^2  \lambda_1(\mathbb{H}) \to 0$ and $\lambda_2(\mathbb{H})>\lambda_1(\mathbb{H})$, we conclude that $b_N\to 0$ as $N\to \infty$. Therefore, $f_M \Phi_N\to \Phi^{(1)}$, and hence $\Phi_N\to \Phi^{(1)}$, as $N\to \infty$.
\begin{remark}[Remark on the convergence in the form domain] \it We can show that if (A3s) holds true, then we have the strong convergence $\Phi_N \to \Phi^{(1)}$ in the norm induced by the quadratic form of $\bH$ on $\cF_+$, namely $\left \langle \Phi_N, \bH \Phi_N \right \rangle \to \left \langle \Phi^{(1)}, \bH \Phi^{(1)} \right \rangle.$ In fact, from the above proof we already had
$$\left \langle f_M \Phi_N, \bH f_M \Phi_N \right \rangle \to \left \langle \Phi^{(1)}, \bH \Phi^{(1)} \right \rangle~{\rm and}~\left\langle {g_M \Phi_N ,\widetilde {H}_N g_M \Phi_N} \right\rangle\to 0.$$
On the other hand, from (A3s) and (\ref{eq:comparison-bH-dGamma(h+1)}), we get $\widetilde {H}_N  \ge c_0 (\bH+\N_+) +g(N)$ on $\F_+^{\le N}$, where $c_0>0$ and $g(N)\to 0$ as $N \to \infty$. By choosing $M$ such that $\max\{g(N), N\eps_R(N), \sqrt{N}\}\ll M \ll N$, we obtain
$$g_M \widetilde {H}_N g_M \ge c_0 (\bH + M) +g(N) \ge c_0 \bH.$$
It implies that $\left \langle g_M \Phi_N, \bH g_M \Phi_N \right \rangle \to 0$. By using the localization for $\bH$ in Proposition \ref{le:localization-H} we can conclude that $\left \langle \Phi_N, \bH \Phi_N \right \rangle \to \left \langle \Phi^{(1)}, \bH \Phi^{(1)} \right \rangle.$ 
\end{remark}

The convergence of excited states of $H_N$ can be proved by using the above argument and the following abstract result.

\begin{lemma}[Convergence of approximate eigenvectors]\label{le:convergence-eigenvectors} Assume that $\A$ is a self-adjoint operator, which is bounded from below, on a (separable) Hilbert space, with the min-max values $\lambda_1(\A)\le ...\le \lambda_L(\A)<\inf \sigma_{\ess}(\A)$. If the normalized vectors $\{x^{(j)}_n\}_{n\ge 1}^{L\ge j \ge 1}$ satisfy, for all $i,j\in \{1,2,...,L\}$,   
$$\lim_{n\to \infty }\langle x_n^{(i)}, x_n^{(j)} \rangle = \delta_{ij}~~\,{\rm and} ~\,\lim_{n\to \infty}\langle x_n^{(j)}, \A x_n^{(j)} \rangle = \lambda_j(\A),$$
then there exists a subsequence $\{x_{n_k}^{(L)}\}_{k\ge 1}$ which converges, in the norm induced by the quadratic form of $\A$, to an eigenvector of $\A$ with the eigenvalue $\lambda_L(\A)$. 
\end{lemma}

The proof of Lemma \ref{le:convergence-eigenvectors} is elementary, using the same argument of proving the convergence of ground states, and an induction process. The proof of Theorem \ref{thm:lower-spectrum} is finished.
\qed

\subsection{Proof of Theorem \ref{thm:positive-temperature}}\label{sec:proof:positive-temp}

\subsubsection*{Step 1. Convergence of the free energy}
We need to show that
$$ \lim_{N\to \infty}(F_\beta (N)-Ne_{\rm H})= -\beta^{-1}\log \Tr_{\F_+} [e^{-\beta \bH}].$$
We can rewrite
$$F_\beta(N)-N e_{\rm H}= \inf_{\Gamma \ge 0, \Tr_{\F_+^{\le N}} (\Gamma)=1} \{ \Tr[
 \widetilde{H}_N \Gamma]-\beta^{-1}S(\Gamma) \}$$
where $\widetilde {H}_N:=U_N (H_N-Ne_{\rm H})U_N^*$, and 
$$ -\beta^{-1}\log \Tr_{\F_+} [e^{-\beta \bH}]= \Tr[\bH \overline{\Gamma}]-\beta^{-1}S(\overline{\Gamma})$$
where $\overline{\Gamma}:= Z^{-1}e^{-\beta \bH}$ with $Z=\Tr \left[ e^{-\beta \bH} \right]$.

\smallskip

\paragraph*{\it Upper bound.} 
Let us write $ \overline{\Gamma}= \sum_{i=1}^\infty t_i \left| \Phi^{(i)} \rangle \langle \Phi^{(i)} \right|$ where $\{\Phi^{(i)}\}_{i=1}^\infty$ is an orthonormal family in $\F_+$ and $t_1 \ge t_2 \ge ... \ge 0$, $\sum t_i=1$. Then
\bq \label{eq:free-energy-Gibbs-state}
\Tr_{\F_+} \left[\bH \overline{\Gamma} \right] -\beta^{-1} S(\overline{\Gamma}) = \sum_{i=1}^\infty \left( t_i \langle \bH \rangle_{\Phi^{(i)}} + \beta^{-1} t_i \log t_i \right).
\eq

Fix $L\in \mathbb{N}$. By using Lemma \ref{le:localization-H} and the fact that $\bH$ is bounded from below, we can find for every $M \ge 1$ a family of normalized states $\{\Phi_M^{(i)}\}_{i=1}^L  \subset \F_+^{\le M}$ such that  $\lim_{M\to \infty} \langle \Phi_M^{(i)}, \Phi_M^{(j)} \rangle =\delta_{ij}$ and
\bq \label{eq:conv-localization}
\limsup_{M\to \infty} \langle \bH \rangle_{\Phi_M^{(i)}} \le \langle \bH \rangle_{\Phi^{(i)}}
\eq for all $i,j\in \{1,2,...,L\}$. Denote $\theta_L=\sum_{i=1}^L t_i$ and
$$ \Gamma_{L,M} := \sum_{i=1}^L \frac{t_i}{\theta_L} \left| \Phi_M^{(i)} \right\rangle \left\langle \Phi_{M}^{(i)} \right|.$$
Then it is easy to see that $\Gamma_{L,M}\ge 0$ and $\Tr[\Gamma_{L,M}]=1$. Moreover, because $ \lim_{M\to \infty}  \langle \Phi_M^{(i)}, \Phi_M^{(j)} \rangle= \delta_{ij}$ we get
\bq \label{eq:conv-entropy-truncated}
 \lim_{M\to \infty} S(\Gamma_{L,M})= - \sum_{i=1}^L \frac{t_i}{\theta_L} \log \left( \frac{t_i}{\theta_L} \right).
 \eq
Choosing $M=N^{1/3}$ and applying Proposition \ref{le:bound-truncated-Fock-space} we obtain
\bq \label{eq:conv-energy-truncated}
  \Tr_{\F_+^{\le N}} [\widetilde {H}_N \Gamma_{L,M}] - \Tr_{\F_+} [\bH \Gamma_{L,M}]   \to 0
\eq
when $N \to \infty$. 
From (\ref{eq:conv-localization}), (\ref{eq:conv-entropy-truncated}) and (\ref{eq:conv-energy-truncated}), we find that 
\bqq
\limsup_{N\to \infty} \left( F_\beta (N)-N e_{\rm H} \right) &\le & \limsup_{N\to \infty} \left( \Tr_{\F_+^{\le N}} [\widetilde {H}_N \Gamma_{L,M}] - \beta^{-1} S (\Gamma_{L,M}) \right) \hfill\\
&\le & \sum_{i=1}^L \left( \frac{t_i}{\theta_L} \langle \bH \rangle_{\Phi^{(i)}} + \beta^{-1} \frac{t_i}{\theta_L } \log \left(\frac{t_i}{\theta_L} \right) \right).
\eqq

Finally, taking  $L\to \infty$ and noting that $\theta_L\to 1$, we obtain from (\ref{eq:free-energy-Gibbs-state}) that
\bqq
\limsup_{N\to \infty} \left( F_\beta (N)-N e_{\rm H} \right) \le  -\beta^{-1} \log \Tr_{\F_+} \left[ e^{-\beta \bH} \right]  .
\eqq

\smallskip

\paragraph*{\it Lower bound.} Let us denote $\Gamma_N:= U_N \Gamma_{\beta,N} U_N^*=e^{-\beta \widetilde {H}_N}/\Tr\left[e^{-\beta \widetilde {H}_N}\right]$. Using $\Tr[\widetilde {H}_N \Gamma_N]-\beta^{-1}S(\Gamma_N) \le 0$ and the stability 
$$\Tr[\widetilde {H}_N \Gamma_N]-\beta_0^{-1} S(\Gamma_N) \ge F_{\beta_0}(N)\ge -CN$$
with $\beta_0^{-1}>\beta^{-1}$, we obtain $\Tr[\widetilde {H}_N \Gamma_N]\le CN$. Therefore, by using Lemma \ref{le:localization-HN}, the simple bound $\bH \le C(\widetilde {H}_N + CN)$, and Proposition \ref{le:bound-truncated-Fock-space} we find that
\bqq
\Tr[\widetilde {H}_N \Gamma_N] &\ge & \Tr[\widetilde {H}_N \Gamma_N^{\le M}] + \Tr[\widetilde {H}_N \Gamma_N^{>M}]-\frac{CN}{M^2} \hfill\\
&\ge & \left( 1-C\sqrt{\frac{M}{N}} \right) \Tr[\bH \Gamma_N^{\le M}] + \Tr[\widetilde {H}_N \Gamma_N^{>M}]-\frac{CN}{M^2} -C\sqrt{\frac{M}{N}}
\eqq
with $\Gamma_N^{\le M}:=f_M \Gamma_N f_M$ and $\Gamma_N^{>M}:=g_M \Gamma_N g_M$. On the other hand, using $f_M^2+g_M^2=1$ and the Brown-Kosaki inequality \cite{BroKos-90}, we have 
$$ S (\Gamma_N)\le S(\Gamma_N^{\le M})+S(\Gamma_N^{>M'}).$$
If we choose $M=N^{3/5}$, then the above estimates imply that 
\bqq
F_\beta (N)-Ne_{\rm H}&=& \left( \Tr[\widetilde {H}_N \Gamma_{N}]-\beta^{-1}S(\Gamma_{N}) \right) \hfill\\
&\ge &  \left( 1- CN^{-1/5} \right) \Tr[\bH \Gamma_N^{\le M}] - \beta^{-1} S(\Gamma_N^{\le M})\hfill\\
&~&+\Tr[\widetilde {H}_N \Gamma_N^{>M} ] -  \beta^{-1}S(\Gamma_N^{>M})- CN^{-1/5}
.
\eqq
By using 
\bqq
\left( 1- \frac{C}{N^{1/5}} \right) \Tr \left[\bH \frac{\Gamma_N^{\le M}}{\Tr \Gamma_N^{\le M}} \right] &-&\beta^{-1} S \left( \frac{\Gamma_N^{\le M}}{\Tr \Gamma_N^{\le M}}\right) \\
&\ge &  -\beta^{-1}\log \Tr e^{-\beta \left(1- CN^{-1/5}\right)\bH}
\eqq
and
$$ \Tr \left[\widetilde {H}_N \frac{\Gamma_N^{>M}}{\Tr \Gamma_N^{>M}} \right] -\beta^{-1} S \left( \frac{\Gamma_N^{>M}}{\Tr \Gamma_N^{>M}}\right) \ge F_\beta (N)-Ne_{\rm H},
$$
we can conclude that
\bqq
F_\beta (N)-Ne_{\rm H}+\frac{\eps}{\Tr \Gamma_N^{\le M}} \ge -\beta^{-1}\log \Tr e^{-\beta (1-CN^{-1/5})\bH}\qquad\qquad\qquad\\
\qquad\qquad\qquad\qquad\qquad\qquad +  \beta^{-1} \log (\Tr \Gamma_N^{\le M}) + \beta^{-1}\frac{\Tr \Gamma_N^{>M'}}{\Tr \Gamma_N^{\le M} } \log (\Tr \Gamma_N^{>M}).
\eqq
Assumption (A3') implies that $\Tr \Gamma_N^{>M} \to 0$ and $\Tr \Gamma_N^{\le M}\to 1$ as $N\to \infty$. Therefore, 
\bqq
 \liminf_{N\to \infty} \left( F_\beta(N)-Ne_{\rm H} \right) &\ge &  \liminf_{N\to \infty}\left( -\beta^{-1}\log \Tr e^{-\beta (1-CN^{-1/5})\bH} \right) \\
&=& -\beta^{-1}\log \Tr e^{-\beta \bH}.
\eqq
Here in the latter equality we have employed the Dominated Convergence Theorem using $\Tr e^{-\beta_0 \bH}<\infty$ for $\beta_0<\beta$ in (A4). 

\subsubsection*{Step 2. Convergence of Gibbs states}

From the above proof, we have  
$$\lim_{N\to \infty}\left( \Tr[ \bH \widetilde {\Gamma}_N]- \beta^{-1}S(\widetilde {\Gamma}_N)\right)=\Tr[\bH \overline{\Gamma}]-\beta^{-1}S(\overline{\Gamma})$$
where 
$$\widetilde {\Gamma}_N:=\frac{f_M \Gamma_N f_M}{\Tr[f_M \Gamma_N f_M]}~~{\rm with}~M=N^{3/5}.$$

We have proved that $\Tr_{\F_+} [U_N e^{-\beta(H_N-N e_{\rm H})} U_N^*] \to  \Tr_{\F_+} e^{-\beta \bH}$ as $N\to \infty$, and we will now show that $\Gamma_N \to \overline{\Gamma}$ in the Hilbert-Schmidt norm. It is well known that a sequence of non-negative operators $A_N$ with $\Tr(A_N)=1$, which converges weakly-$\ast$ to an operator $A$, converges in the trace norm if and only if $\Tr(A)=1$ (see, e.g.,~\cite{dellAntonio-67},~\cite[Cor. 1]{Robinson-70} and~\cite[Add. H]{Simon-79}). Using this fact, we will get the result. Moreover, by using $\Tr [(\widetilde {\Gamma}_N-\Gamma_N)^2] \to 0$  due to the condensation (A3'), it remains to show that $\Tr [(\widetilde {\Gamma}_N-\overline{\Gamma})^2]\to 0$. The latter convergence follows from the equality 
$$ \Tr[\bH \widetilde {\Gamma}_N]- \beta^{-1}S(\widetilde {\Gamma}_N)-\left(\Tr[\bH \overline{\Gamma}]-\beta^{-1}S(\overline{\Gamma}) \right) = \beta^{-1} \Tr \left[ \Gamma_N (\log \widetilde {\Gamma}_N - \log \overline \Gamma) \right] $$
and the following entropy estimate, which is inspired from \cite[Theorem 1]{HaiLewSei-08}. The proof of Theorem \ref{thm:positive-temperature} is finished. \qed

\begin{lemma}[Relative entropy inequality] \label{le:entropy-inequality} If $A$ and $B$ are two trace class operators on a Hilbert space and $0\le A \le 1$, $0\le B \le 1-\eps$ for some $\eps>0$, then there exists $C_\eps>0$ such that 
$$ \Tr \left[ A (\log A -\log B) \right] \ge C_\eps \Tr [(A-B)^2]+\Tr(A-B).$$
\end{lemma}

\begin{proof} A straightforward computation shows that for all $x,y\in [0,1]$,
$$ x(\log x - \log y)-(x-y) \ge g(y) (x-y)^2~~{\rm with}~g(y):=\frac{y-1-\log y}{(y-1)^2}.$$
Since the function $y\mapsto g(y)$ is decreasing, we have $g(y)\ge g(1-\eps)>0$ for all $y\in [0,1-\eps]$. Therefore, By Klein's inequality \cite[p. 330]{Thirring}, one has
$$\Tr \left[ A (\log A - \log B) \right] \ge g(1-\eps)\Tr[(A-B)^2]+ \Tr(A-B).$$
The proof of Lemma \ref{le:entropy-inequality} is finished.
\end{proof}

\subsubsection*{Remark on stability and condensation}

We prove the remark that if (A3s) holds true for some $\eps_0\in (0,1)$ and $\Tr_{\F_+} \left[ e^{-(1-\eps_0)\beta_0 \bH } \right] <\infty$ for some $\beta_0$, then (A3') and (A4) hold true.

In fact, using $\mathbb{H}\le d\Gamma(h+C)+C$ and $\Tr_{\F_+} \left[ e^{-(1-\eps_0)\beta_0 \bH } \right] <\infty$ we obtain that $\Tr_{\F_+} \left[ e^{-(1-\eps_0)\beta_0 {\rm d}\Gamma (h+C) } \right] <\infty$. Note that for every positive operator $A$ in $\gH_+$, we have $\Tr_{\F_+} e^{-\dGamma(A)}<\infty$ if and only if $\Tr_{\gH_+} e^{-A}<\infty$. Therefore, we can conclude that $\Tr e^{-(1-\eps_0) \beta_0 h}<\infty$ and $\Tr_{\F_+} e^{-(1-\eps_0) \beta_0 \dGamma(h)}<\infty$. From the latter bound and the inequality $\widetilde {H}_N \ge (1-\eps_0) \sum_{i=1}^N h_i +o(N)$ in (A3s), the stability follows:
\bqq
F_{\beta_0}(N)-N e_{\rm H} &=& -\beta_0^{-1} \log \Tr_{\F_+^{\le N}} \left[e^{-\beta_0 \widetilde {H}_N} \right] \hfill\\
&\ge & -\beta_0^{-1} \log \Tr_{\F_+^{\le N}} \left[e^{-(1-\eps_0) \beta_0 \sum_{j=1}^N h_j }\right] + o(N) \hfill\\
&\ge & -\beta_0^{-1} \log \Tr_{\F_+} \left[e^{-(1-\eps_0) \beta_0 \dGamma(h)} \right] + o(N) \ge -C+ o(N).
\eqq
Next, using $\widetilde {H}_N \ge c_0 \N_+ +o(N)$ on $\F_+^{\le N}$ and the stability
$$\Tr[\widetilde {H}_N \widetilde {\Gamma}_{\beta,N}]-\beta_0^{-1} S(\widetilde {\Gamma}_{\beta,N}) \ge F_{\beta_0}(N)-N e_{\rm H} \ge  o(N)$$
with $\widetilde {\Gamma}_{\beta,N}:= U_N {\Gamma}_{\beta,N} U_N^*$, we find that
\bqq 0 &\ge &  F_\beta(N)-N e_{\rm H}= \Tr[\widetilde {H}_N \widetilde {\Gamma}_{\beta,N}]-\beta^{-1} S(\widetilde {\Gamma}_{\beta,N}) \hfill\\
&=& \left( 1- \frac{\beta_0}{\beta} \right) \Tr[\widetilde {H}_N \widetilde {\Gamma}_{\beta,N}] +  \frac{\beta_0}{\beta} \left[ \Tr[\widetilde {H}_N \widetilde {\Gamma}_{\beta,N}]-\beta_0^{-1} S(\widetilde {\Gamma}_{\beta,N}) \right] \hfill\\
&\ge & \left( 1- \frac{\beta_0}{\beta} \right) c_0 \Tr[\N_+ \widetilde {\Gamma}_{\beta,N}]+ o(N).
\eqq
Therefore, $\Tr[\N_+ {\Gamma}_{\beta,N}]=\Tr[\N_+ \widetilde {\Gamma}_{\beta,N}]=o(N)$.

\appendix


\section{Bogoliubov Hamiltonian}\label{apd:Bogoliubov-Hamiltonian}

In this appendix we prove some general properties of quadratic Hamiltonians in bosonic Fock spaces, which includes the results in Theorem \ref{thm:Bogoliubov-Hamiltonian}. Most results of this section are well known~\cite{Berezin-66}. In some parts of the discussion we shall refer to \cite{Solovej-notes,Nam-thesis} for more details.

Let $\cH$ be an arbitrary (separable) one-particle Hilbert space. We can identify $\cH$ with its dual space $\cH^*$ by the conjugate linear unitary operator $J: \cH \to \cH^*$ defined as $ (Jf)(g) = \langle f,g \rangle_{\cH}$ for all $f,g\in \cH$. If $\cH$ is a subspace of $L^2(\Omega)$, then $J$ is simply the complex conjugate.

For any vector $f \in \cH$, we may define the {\it annihilation operator} $a(f)$ and the {\it creation operator} $a^*(f)$ on the Fock space $ \F(\cH)=\bigoplus_{N=0}^\infty \bigotimes_{\rm sym}^N \cH$ as in (\ref{eq:def-annihilation-operator}) and (\ref{eq:def-creation-operator}) in  Section \ref{sec:operators}. For any state $\Phi$ in the Fock space $\F(\cH)$, the one-body density matrices $(\gamma_\Phi,\alpha_\Phi)$ are operators on $\cH$ defined by 
\[\left\langle {f,{\gamma _\Phi }g} \right\rangle  = {\left\langle {{a^*}(g)a(f)} \right\rangle _\Phi },\quad\left\langle {f,{\alpha _\Phi } J g } \right\rangle  = {\left\langle {a(g)a(f)} \right\rangle _\Phi }\quad{\rm for~all~}f,g\in \cH.\]

We will be particularly interested in states $\Phi$'s with finite particle expectation, namely $\Tr \gamma_\Phi<\infty.$ In this case, $(\gamma_\Phi,\alpha_\Phi)$ belongs to the set ${\G}$, which contains all pairs of operators $(\gamma ,\alpha)$ on $\cH$ such that $\gamma$ is trace class, $\alpha$ is Hilbert-Schmidt, $\alpha^T=\alpha$ and 
\bqq
\left(  {\begin{array}{*{20}c}
   \gamma  & \alpha   \\
   {\alpha^*} & {1 + J\gamma J^*}  \\

 \end{array} } \right) \geqslant 0 ~~{\rm on}~\cH\oplus \cH^*.
\eqq
Let us denote by $\G_0$ the set of all pairs $(\gamma,\alpha)\in \G$ which satisfy
\bq \label{eq:Gamma-pure}
\alpha \alpha^*=\gamma(1+ J{\gamma} J^*)\quad \text{and}\quad \gamma \alpha=\alpha J{\gamma} J^*.
\eq
The significance of $\G_0$ is that for any $(\gamma,\alpha)\in \G_0$, there exists a unique {\it quasi-free} pure state $\Phi\in \F(\cH)$ such that 
$$(\gamma,\alpha)=(\gamma_\Phi,\alpha_\Phi).$$
Any element in $\G$ is also associated with a unique state, but the latter is a {\it mixed} state. 

The one-body density matrices offer a simple way to define quadratic Hamiltonians. More precisely, let $H$ be a self-adjoint operator on $\cH$ and let $K$ be a Hilbert-Schmidt operator on $\cH$ such that $K=K^{T}$ and such that the following inequality holds true
\bq \label{eq:hK>=eta-general}
\A:=\left( {\begin{array}{*{20}{c}}
  {H}&K \\ 
  K^*&{J H J^*} 
\end{array}} \right) \geqslant \eta>0\qquad \text{on}\quad \cH\oplus \cH^*.
\eq
We shall consider the quadratic Hamiltonian $\mathbb{H}$ on $\F(\cH)$ defined by
\bq \label{eq:quadratic-Hamiltonian}
 \left\langle { \mathbb{H}} \right\rangle_\Phi :=q(\gamma_\Phi, \alpha_\Phi)=\Tr[ H \gamma_\Phi]+\Re \Tr[K\alpha_\Phi]
 \eq
 for every state $\Phi$ living in the truncated Fock spaces and in the domain of $H$:
\begin{equation}
\bigcup_{M\geq0}\bigoplus_{n=0}^M\bigotimes_{\text{sym}}^n D(H).
\label{eq:core-quadratic-Bogoliubov-general} 
\end{equation}
It can be verified that the so-defined operator $\bH$ is exactly the Bogoliubov Hamiltonian given in (\ref{eq:Bogoliubov}). 

The main properties of the quadratic Hamiltonian $\bH$ are given in the following

\begin{theorem}[Bogoliubov Hamiltonian]\label{thm:quadratic-Hamiltonian} Let $\mathbb{H}$ be defined by (\ref{eq:quadratic-Hamiltonian}) and assume that (\ref{eq:hK>=eta-general}) holds true. Then we have the following statements.
\smallskip

\noindent $(i)$ \emph{(Form domain)}. We have the quadratic-form inequalities
\bq \label{eq:comparison-bH-dGamma(h+1)-general}
C^{-1}\dGamma (H) -C \le \bH \le \dGamma (H+C) +C
\eq
As a consequence, the form domain of the Friedrichs extension of $\mathbb{H}$ (still denoted by $\bH$) is the same as that of $\dGamma(H)$ on $\cF_+$.

\smallskip

\noindent $(ii)$ \emph{(Variational principle)}. For any $(\gamma,\alpha)\in \G$, we can find $(\gamma',\alpha')\in \G_0$ such that $ q(\gamma' ,\alpha' )\le q(\gamma,\alpha)$ and the inequality is strict expect when $(\gamma,\alpha)\in \G_0$. As a consequence, the ground state energy of $\mathbb{H}$ is 
$$
\inf \sigma (\mathbb{H})= \inf_{ (\gamma,\alpha)\in \G } q(\gamma , \alpha) =\inf_{ (\gamma,\alpha)\in \G_0 } q(\gamma , \alpha).
$$

\smallskip

\noindent $(iii)$ \emph{(Ground state and ground state energy)}. The Hamiltonian $\mathbb{H}$ has a unique ground state in $\F(\cH)$, which is a pure quasi-free state. Moreover, we always have $\inf \sigma (\mathbb{H})<0$ except when $K=0$ in which case $\inf \sigma (\mathbb{H})=0$ with the vacuum being the corresponding ground state.

\smallskip

\noindent $(iv)$ \emph{(Spectrum)}. We have $ \sigma_{\rm ess}(\mathbb{H}) = \sigma(\mathbb{H})+ \sigma_{\rm ess}(H)$ and 
$$
 \sigma(\mathbb{H})=\inf \sigma(\mathbb{H})+\left\{  \sum_{i\ge 1} {n_i \lambda_i}\,\, |\,\, \lambda_i \in \sigma(\S\A)\cap\R^+, n_i \in \{0\}\cup \mathbb N~{\rm and}~ \sum {n_i}<\infty \right\}
$$
where
$$
\S:=\left( {\begin{array}{*{20}{c}}
  1&0 \\ 
  0&{ - 1} 
\end{array}} \right).
$$
\smallskip

\noindent $(v)$ \emph{(Lower spectrum)}. If $JKJ=K$, $JHJ^*=H$ and $H$ has infinitely many eigenvalues below its essential spectrum, then $\mathbb{H}$ also has infinitely many eigenvalues below its essential spectrum. 

On the other hand, if $JKJ=K \ge 0$, $JHJ^*=H$ and $H-K$ has only finitely  many eigenvalues below its essential spectrum, then $\mathbb{H}$ also has finitely many eigenvalues below its essential spectrum.

\end{theorem}

\begin{remark}\label{rmk:bH-CdGamma(h)} \it Theorem \ref{thm:Bogoliubov-Hamiltonian} follows from Theorem \ref{thm:quadratic-Hamiltonian}, with $H=h+K_1$ and $K=K_2$. In particular, (\ref{eq:comparison-bH-dGamma(h+1)}) follows from (\ref{eq:comparison-bH-dGamma(h+1)-general}) because 
$ C^{-1}(h+1)\le h+K_1 \le C(h+1).$ Moreover, by following the proof of (\ref{eq:comparison-bH-dGamma(h+1)-general}) we also have $\bH \ge \dGamma(h-C)-C$; and if $K_1=K_2\ge 0$, then $\bH \ge d\Gamma(h-\eps)-C_\eps$ for every $\eps>0$.
\end{remark}

\begin{proof} 1. The variational principle is well-known (see, e.g., \cite[Theorem 1.7 p. 101]{Nam-thesis} for a proof). If $K=0$, then from  and the inequality $H\ge \eta>0$, we see that $\inf \sigma(\mathbb{H})=0$ and the vacuum is the unique ground state of $\mathbb{H}$ (which corresponds to $\gamma_\Phi=0$ and $\alpha_\Phi=0$). On the other hand, if $K\ne 0$, then by taking a normalized vector $v\in \cH$ such that $ \langle Jv , Kv \rangle >0$ and then choosing the trial operators 
$$\gamma_\lambda := \lambda |v\rangle \langle v|~~,\alpha_{\lambda}= -\sqrt{\lambda(1+\lambda)} |v \rangle \langle J v|$$
with $\lambda>0$ small, we can see that $\inf \sigma(\mathbb{H})<0$.

Now we show that $\mathbb{H}$ is bounded from below. If $(\gamma,\alpha)\in \G_0$, then we can find an orthonormal family $\{u_n\}_{n\ge 1}$ for $\cH$ such that
$$ \gamma(x,y)=\sum_{n\ge 1}\lambda_n u_n(x)\overline{u_n(y)},~~\alpha(x,y)=\sum_{n \ge 1} \sqrt{\lambda_n(1+\lambda_n)}u_n(x)u_n(y)$$
where $\lambda_n\ge 0$ and $\sum_{n=1}^\infty \lambda_n=\Tr \gamma <\infty$. Thus we can rewrite
$$ q(\gamma,\alpha)=\sum_{n=1}^\infty \left( \lambda_n\left\langle {u_n,H u_n} \right\rangle + \sqrt{\lambda_n(1+\lambda_n)}\Re \left\langle {\overline{u_n},K u_n} \right\rangle \right).$$ 
The inequality (\ref{eq:hK>=eta-general}) implies that $ \left\langle {{u_n}, H u_n} \right\rangle \ge |\left\langle {\overline{u_n},K u_n} \right\rangle| + \eta$ for all $n\ge 1$. Therefore,
\bq \label{eq:Hq-bounded-below}
 q(\gamma,\alpha) &\ge & \eta \Tr\gamma + \sum_{n=1}^\infty \left( \lambda_n-\sqrt{\lambda_n(1+\lambda_n)} \right)|\left\langle {\overline{u_n},K u_n} \right\rangle| \nn\hfill\\
 &\ge & \eta \Tr\gamma - \sum_{n=1}^\infty \sqrt{\lambda_n}|\left\langle {\overline{u_n},K u_n} \right\rangle| \nn\hfill\\
 &\ge &  \eta \Tr\gamma - \sqrt{\sum_{n=1}^\infty \lambda_n} \sqrt{\sum_{n=1}^\infty |\left\langle {\overline{u_n},K u_n} \right\rangle|^2} \nn\hfill\\
 & \ge & \eta \Tr\gamma - \sqrt{\Tr \gamma}  ||K ||_{\rm HS},
 \eq
where $||K||_{\rm HS}$ is the Hilbert-Schmidt norm of $K$. Thus $\bH \ge -C$.

2. Next, we show that $ C^{-1} \dGamma (H) -C \le \bH \le  \dGamma (H+C) +C $. In fact, from the above result, we see that the quadratic Hamiltonian with $\A$ replaced by  
\[{\left( {\begin{array}{*{20}{c}}
  {{H-\eta /2}}&K \\ 
  K^*&{{J H J^* -\eta /2}} 
\end{array}} \right)} \ge \eta /2
\]
is also bounded from below. Therefore, 
$$
\left\langle { \mathbb{H}} \right\rangle_\Phi = (\eta/2) \Tr[\gamma_\Phi]+ \Tr [(H-\eta/2)\gamma_{\Phi}] + \Re \Tr[K \alpha_{\Phi}] \ge (\eta/2) \Tr[\gamma_\Phi] -C.
$$
Similarly, for a constant $C_0>0$ large enough, one has 
$$ 
\Tr [C_0\gamma_{\Phi}] + \Re \Tr[K \alpha_{\Phi}] \ge -C~~{\rm and}~\Tr [ C_0 \gamma_{\Phi}] - \Re \Tr[K \alpha_{\Phi}] \ge -C.
$$
These estimates yield the desired inequalities.
 
3. Now we show that $\bH$ has a ground state.  Let a sequence $\{(\gamma_n,\alpha_n)\}_{n=1}^\infty\subset \G_0$ such that $$\lim_{n\to \infty}q(\gamma_n,\alpha_n)= \inf \sigma(\mathbb{H}).$$
The inequality (\ref{eq:Hq-bounded-below}) implies that $\Tr \gamma_n$ is bounded, and hence $\Tr (\alpha_n \alpha_n^*)$ is also bounded. Thus up to a subsequence, we may assume that there exists $(\gamma_0,\alpha_0)\in \G$ such that $\alpha_n \wto \alpha_0~~{\rm and}~\gamma_n \wto \gamma_0 $ weakly in the Hilbert-Schmidt norm. Consequently, $ \lim_{n\to \infty}\Tr[K \alpha_n ]= \Tr[K \alpha_0]$ and $\liminf_{n\to \infty}\Tr[H\gamma_n ] \ge \Tr[H \gamma_0]$ by Fatou's lemma since $H\ge 0$. 
Thus $(\gamma_0,\alpha_0)$ is a minimizer of $q(\gamma,\alpha)$ on $\G$. Due to (i), this minimizer $(\gamma_0,\alpha_0)$  belongs to $\G_0$. 

4. To understand the structure of one-body densities matrices and the spectrum of the quadratic Hamiltonian, let us introduce {\it Bogoliubov transformations}. A Bogoliubov transformation $\V$ is a {linear} bounded isomorphism on $\cH\oplus \cH^*$ such that $(\V\V^*-1)$ is trace class ({\it Stinespring condition}) and
\[ \V \left( {\begin{array}{*{20}{c}}
  0&{{J^*}} \\ 
  J&0 
\end{array}} \right) = \left( {\begin{array}{*{20}{c}}
  0&{{J^*}} \\ 
  J&0 
\end{array}} \right) \V, \quad \V \S \V^* = \S :=\left( {\begin{array}{*{20}{c}}
  1&0 \\ 
  0&{ - 1} 
\end{array}} \right).\]

Since $(\gamma_0,\alpha_0)\in \G_0$, there exists a Bogoliubov transformation $\V_0$ on $\cH \oplus \cH^*$ which diagonalizes the one-body density matrices $(\gamma_0,\alpha_0)$, namely
\[ \V_0 \left( {\begin{array}{*{20}{c}}
  \gamma_0 &{{\alpha_0}} \\ 
  \alpha_0^* & 1+ J\gamma_0 J^* 
\end{array}} \right){\V_0^*} = \left( {\begin{array}{*{20}{c}}
  0&{{0}} \\ 
  0&1 
\end{array}} \right)\]
(see, e.g., \cite{Solovej-notes,Nam-thesis}). Then by employing the fact that $(\gamma_0,\alpha_0)$ is a minimizer for $q(\gamma,\alpha)$ on $\G$, we can show (see \cite[Theorem 1.7 p. 101]{Nam-thesis} for a detailed proof) that the Bogoliubov transformation $\V:=\S \V_0 \S$ diagonalizes $\A$, namely 
\bq \label{eq:Bogoliubov-diagonalization} 
\V  \A {\V^*} = \left( {\begin{array}{*{20}{c}}
  \xi &{0} \\ 
  0& J\xi J^* 
\end{array}} \right)
\eq
for some operator $\xi$ on $\cH$. From (\ref{eq:hK>=eta-general}) and (\ref{eq:Bogoliubov-diagonalization}), we deduce that $\xi \ge \eta >0$. 

Using the Bogoliubov transformation $\V$, we can find a unitary transformation $\U_{\V}$ on the Fock space $\F(\cH)$, called the Bogoliubov unitary (see e.g. \cite{Solovej-notes,Nam-thesis}), such that
$$ \U_\V \mathbb{H} \U_{\V}^*-\inf \sigma(\mathbb{H}) = \dGamma(\xi) := \bigoplus_{N=0}^\infty \sum_{i=1}^N \xi_i.$$
This implies that $\mathbb{H}$ has a unique ground state and $\sigma(\mathbb{H})=\inf \sigma(\mathbb{H})+\dGamma (\xi)$.

5. Next, from (\ref{eq:Bogoliubov-diagonalization}) and the assumption $\V\S\V^*=S$, it follows that 
\[
\S \left( {\begin{array}{*{20}{c}}
  \xi &0 \\ 
  0&{J\xi {J^*}} 
\end{array}} \right) - \lambda  = \S \V \A {\V ^*} - \lambda  = \S \V \S(\S \A - \lambda ){\V ^*}.\]
Thus 
$$\sigma(\xi)=\sigma \left( \S \left( {\begin{array}{*{20}{c}}
  \xi &0 \\ 
  0&{J\xi {J^*}} 
\end{array}} \right) \right) \cap \R^+ =\sigma(\S\A)\cap \R^+$$
because $\V$ is an isomorphism on $\cH\oplus \cH^*$. Consequently,
$$
 \sigma(\mathbb{H})=\inf \sigma(\mathbb{H})+\left\{  \sum_{i\ge 1} {n_i \lambda_i}\,\, |\,\, \lambda_i \in \sigma(\S\A)\cap\R^+, n_i \in \{0\}\cup \mathbb N~{\rm and}~ \sum {n_i}<\infty \right\}.
$$

6. Now we prove that $\sigma_{\rm ess}(\xi)=\sigma_{\rm ess}(H)$. From (\ref{eq:Bogoliubov-diagonalization}), it follows that $\sigma(\xi)=\sigma(\V \A \V^*)$. On the other hand, since $K$ is a compact operator on $\cH$, one has $\sigma_{\rm ess}(\A)= \sigma_{\rm ess}(H).$ On the other hand, by using the identity
$$ \V\A\V^*-\lambda = \V(\A-\lambda)\V^*+\lambda (\V \V^* -1) $$
and the fact that $(\V\V^*-1)$ is compact (indeed it is trace class), we obtain $ \sigma_{\rm ess} (\V\A\V^*)=\sigma_{\rm ess}(\A).$ Thus $\sigma_{\rm ess}(\xi)=\sigma_{\rm ess}(H)$. Consequently, from $\sigma(\mathbb{H})=\inf \sigma(\mathbb{H})+\sigma( \dGamma (\xi))$ we obtain
\bqq 
 \sigma_{\ess}(\mathbb{H})=\inf \sigma(\mathbb{H})+\sigma_{\rm ess}( \dGamma (\xi))=\sigma(\mathbb{H}) + \sigma_{\rm ess}(H).
\eqq 

7. From now on, we assume that $JKJ=K$ and $JHJ^*=H$. In this case, (\ref{eq:hK>=eta-general}) implies that $H+K\ge \eta>0$ and $H-K\ge \eta>0$. Before going further, let us give a simple characterization of the spectrum of $\bH$. We only deal with eigenvalues for simplicity.

Assume that $t > 0$ is an eigenvalue of $\xi$. Then from the above proof, we see that $t$ is an eigenvalue of $\S \A$. Using the equation 
\[0 = (\S\A - t )\left( \begin{gathered}
  u \hfill \\
  v \hfill \\ 
\end{gathered}  \right) = \left( {\begin{array}{*{20}{c}}
  H&K \\ 
  { - K}&{ - JH{J^*}} 
\end{array}} \right)\left( \begin{gathered}
  u \hfill \\
  v \hfill \\ 
\end{gathered}  \right) - \lambda \left( \begin{gathered}
  u \hfill \\
  v \hfill \\ 
\end{gathered}  \right)\]
we find that
\[\left\{ \begin{gathered}
  (H + K)x = t y, \hfill \\
  (H - K)y = t x ,\hfill \\ 
\end{gathered}  \right.
\]
where $ x=u+v$ and $y=u-v.$ Thus
\[
\left( {H + K - \frac{{{t^2}}}{{H - K}}} \right) x =0.
\]
Note that $x\ne 0$ if $u\oplus v \ne 0$. Therefore, $0$ is an eigenvalue of 
$$ \mathcal{X}_t := H+K-\frac{t^2}{H-K}.$$
In fact, using Weyl sequences and arguing similarly, we also obtain 
\bq \label{eq:equivalence-spectrum-xi-SA-Xt}
t \in \sigma(\xi) = \sigma(\S \A) \Leftrightarrow 0\in \sigma (\mathcal{X}_t).
\eq

8. We show that the number of eigenvalues (with multiplicity) of $\xi$  below $\mu:=\inf \sigma_{\rm ess} (\xi)$ is equal to the number of negative eigenvalues of $\mathcal{X}_\mu$.

Note that the mapping $t\mapsto \mathcal{X}_t$ is strictly decreasing. As a consequence, for every $j=1,2,...$, the min-max value $\lambda_j(\mathcal{X}_t)$ is a continuous and decreasing function on $t\ge 0$. More precisely, if $0\le t_1<t_2$, then $\lambda_j(\mathcal{X}_{t_1}) \ge \lambda_j(\mathcal{X}_{t_2})$ and the inequality is strict if $\lambda_j(\mathcal{X}_{t_1})$ is an eigenvalue of $\mathcal{X}_{t_1}$.  Moreover, 
$$\inf\sigma(\mathcal{X}_0)= \inf\sigma(H+K)\ge \eta>0$$
and, for every $t\in (0,\mu)$, 
$$ \inf \sigma_{\rm ess} (\mathcal{X}_t) > \inf \sigma_{\rm ess}(\mathcal{X}_\mu) = 0.$$
Therefore, by using (\ref{eq:equivalence-spectrum-xi-SA-Xt}), we obtain a one-to-one correspondence
\bq \label{eq:spectrum-xi-Xt-one-to-one}
\sigma(\xi)\cap (-\infty,\mu) \leftrightarrow \sigma(\mathcal{X}_\mu) \cap (-\infty,0).
\eq

9. From the inequality $\mu^2(H - K)^{-1} + H - K \geqslant 2\mu$ we get
\[ \mathcal{X}_\mu = H + K - \frac{{{\mu^2}}}{{H - K}} \leqslant 2 \left( H - \mu \right). \]
Moreover, note that $\inf \sigma_{\rm ess}(\mathcal{X}_\mu)=\inf \sigma_{\rm ess} (H - \mu)=0.$
Therefore, the number of negative eigenvalues (with multiplicity) of $\mathcal{X}_\mu$  is not less than the number of negative eigenvalues (with multiplicity) of $(H-\mu)$ .

In particular, if $(H-\mu)$ has infinitely many negative eigenvalues, then $\mathcal{X}_\mu$ has infinitely many negative eigenvalues. By (\ref{eq:spectrum-xi-Xt-one-to-one}), we see that $\xi$ has infinitely many eigenvalues below its essential spectrum. Consequently, $\mathbb{H}$ has infinitely many eigenvalues below its essential spectrum.   

10. Now we assume furthermore that $K\ge 0$. Then we get the inequality
\[ \mathcal{X}_\mu = H + K - \frac{{{\mu^2}}}{{H - K}} \ge H-K -\frac{\mu^2}{H-K}. \]
Moreover, note that 
$$\inf \sigma_{\rm ess}(\mathcal{X}_\mu)=\inf \sigma_{\rm ess} \left( H-K -\frac{\mu^2}{H-K} \right) =0.$$
Thus, if $H-K$ has finitely many eigenvalues below $\mu$, then $\mathcal{X}_t$ also has finitely many negative eigenvalues. By (\ref{eq:spectrum-xi-Xt-one-to-one}), $\xi$ has finitely many eigenvalues below its essential spectrum. Consequently, $\mathbb{H}$ has the same property.
\end{proof}

\section{Localization of band operators on $\F_+$}\label{apd:localization}

In this appendix we prove the localization in Proposition \ref{pro:localization}.

\begin{proof}[Proof of Proposition \ref{pro:localization}] Using the IMS-identity
$$\A= f_M\A f_M + g_M \A g_M + \frac{1}{2} \left( [f_M,[f_M,\A]]+ [g_M,[g_M,\A]] \right)$$
we can write, with $\Phi=\oplus_{j=0}^\infty \Phi_j\in \oplus_{j=0}^\infty \gH_+^j =\F_+$,
\[\begin{gathered}
   \left\langle {{\Phi},\A{\Phi}} \right\rangle   = \left\langle {f_M\Phi,\A f_M\Phi } \right\rangle + \left\langle {g_M\Phi,\A g_M\Phi } \right\rangle\hfill \\
\qquad\qquad  + \sum\limits_{1\le |i-j | \le \sigma} {\left[ {{{\left( {f_M\left( {i} \right) - f_M\left( {j} \right)} \right)}^2} + {{\left( {g_M\left( {i} \right) - g_M\left( {j} \right)} \right)}^2}} \right] }  \pscal{\Phi_i,\A\Phi_j}. \hfill \\ 
\end{gathered} \]
Since $f$ and $g$ are smooth, we have
$$ {{{\left( {f_M\left( {i} \right) - f_M\left( {j} \right)} \right)}^2} + {{\left( {g_M\left( {i} \right) - g_M\left( {j} \right)} \right)}^2}} \le \left( \|f'\|^2_{\infty}+\|g'\|^2_{\infty}\right)\frac{(i-j)^2}{M^2}. $$
Moreover, using the assumption that $\A \ge 0$ we get 
\bqq
\sum\limits_{1 \leqslant |i-j| \leqslant \sigma} |\pscal{\Phi_i,\A\Phi_j}|  &\leqslant & \frac{1}{2}\sum\limits_{1 \leqslant |i-j| \leqslant \sigma}   \left(\pscal{\Phi_i,\A\Phi_i} + \pscal{\Phi_j,\A\Phi_j}\right) \hfill\\
&\le & 2\sigma \sum\limits_i \pscal{\Phi_i,\A\Phi_i} = 2\sigma   \left\langle \Phi ,\A_0 \Phi \right\rangle .
\eqq
Therefore, we obtain the operator inequality (\ref{eq:IMS-band-operator}).

Finally, let us show that if 
$$\delta:=\sup \{||g_M\Phi ||^2: \Phi\in Y, ||\Phi||=1\}< (\dim Y)^{-1},$$
then $\dim(f_M Y)=\dim Y$. In fact, assume that $\dim Y=L$ and let  $\{\Phi_i\}_{i=1}^L$ be an  orthonormal basis for $Y$. For all $\{\alpha_i\}_{i=1}^L \in \mathbb{C}^L \minus \{0\}$, we have
\[
\begin{gathered}
  \left\| \sum\limits_{i = 1}^L \alpha _i f_M \Phi _i \right\|^2
\ge \sum\limits_{i = 1}^L |\alpha _i|^2 \left\| f_M\Phi _i \right\|^2  - 2\sum\limits_{1 \leqslant i < j \leqslant L} \left| \alpha _i\alpha _j \right|  \cdot  \left| \left\langle f_M\Phi _i,f_M\Phi _j\right\rangle  \right|  \hfill \\
= \sum\limits_{i = 1}^L |\alpha _i|^2 \Big( 1- \left\| g_M\Phi _i \right\|^2 \Big)  - 2\sum\limits_{1 \leqslant i < j \leqslant L} {\left| {{\alpha _i\alpha _j}} \right|  \cdot \left| {\left\langle {g_M{\Phi _i},g_M{\Phi _j}} \right\rangle } \right|}  \hfill \\
   \ge \sum\limits_{i = 1}^L {|{\alpha _i}{|^2}\left( {1 - {\delta}} \right)}  - \sum\limits_{1 \leqslant i < j \leqslant L} {\big( {{{\left| {{\alpha _i}} \right|}^2} + {{\left| {{\alpha _j}} \right|}^2}} \big){\delta}}  = (1 - L{\delta})\sum\limits_{i = 1}^L {|{\alpha _i}{|^2}}  > 0.
\end{gathered} \]
Therefore, the subset $\{f_M \Phi_j \}_{j=1}^L\subset f_M Y$ is linearly independent, which implies that $\dim (f_M Y) \ge L=\dim Y$. On the other hand, $\dim(f_M Y)\le \dim Y$ because $f_M$ is linear. Thus $\dim(f_M Y)=\dim Y$.
\end{proof}

\section{Logarithmic Lieb-Oxford inequality}\label{apd:log-Lieb-Oxford-inequality}

We provide a proof of Proposition \ref{le:2d_LO}. We follow ideas from~\cite{Lieb-79,LieOxf-80,LieSei-09}.

\begin{proof}[Proof of Proposition \ref{le:2d_LO} ] Let us write for simplicity $\rho:=\rho_\Psi$. Note that for any $\mu_{x_i}\geq0$ which is radially symmetric about $x_i$ and such that $\int\mu_{x_i}=1$, we have
\begin{align*}
0  &\le   \frac12 D\left(\rho-\sum_{i=1}^N\mu_{x_i},\rho-\sum_{i=1}^N\mu_{x_i}\right)\\
&=  \frac12 D(\rho,\rho)+\sum_{1\leq i<j\leq N}D(\mu_{x_i},\mu_{x_j})+\frac12\sum_{i=1}^ND(\mu_{x_i},\mu_{x_i})-\sum_{i=1}^ND(\rho,\mu_{x_i}).\\
\end{align*}
We have used here Proposition~\ref{pro:log-kernel} together with the fact that $\int_{\R^2}(\rho-\sum_{i=1}^N\mu_{x_i})=0$.
By Newton's theorem, $(\mu_{x_i}\ast w)(x)\leq w(x-x_i)$ a.e., and therefore
$$D(\mu_{x_i},\mu_{x_j})\leq w(x_i-x_j)$$
for any $i\neq j$. We arrive at the following Onsager-type estimate~\cite{Onsager-39}
\begin{equation}
\sum_{1\leq i< j\leq N}w(x_i-x_j)\geq -\frac12 D(\rho,\rho)+\sum_{i=1}^ND(\rho,\mu_{x_i})-\frac12\sum_{i=1}^ND(\mu_{x_i},\mu_{x_i}).
\label{eq:Onsager}
\end{equation}
Writing $D(\rho,\mu_{x_i})=(\rho\ast w)(x_i)+D(\rho,\mu_{x_i}-\delta_{x_i})$ and taking the expectation value against $\Psi$, we get
\begin{multline*}
\pscal{\sum_{1\leq i<j\leq N}w(x_i-x_j)}_\Psi\geq \frac12 D(\rho,\rho)
+\int_{\R^2}\rho(x)\,D(\rho,\mu_{x}-\delta_x)\,\d x\\-\frac12\int_{\R^2}\rho(x)\,D(\mu_{x},\mu_{x})\,\d x. 
\end{multline*}
Now we choose 
$$\mu_x(y)=R(x)^{-2}\,\mu\left(\frac{y-x}{R(x)}\right)\quad\text{and}\quad R(x)=\sqrt{\frac{\lambda(x)}{\rho(x)}},$$
where $\mu= \pi^{-1} \chi(|x|\le 1)$ is the (normalized) characteristic function of the unit ball. We will choose the function $\lambda(x)$ at the very end.

We start by computing
$$D(\mu_{x},\mu_{x})=-\int_{\R^2}\!\int_{\R^2}\mu(y)\mu(z)\log\big(R(x)|y-z|\big)\d y\,dz=\frac12\log\frac{\rho(x)}{\lambda(x)}+D(\mu,\mu)$$
which gives
$$-\frac {1}{2}\int_{\R^2}\rho(x)\,D(\mu_{x},\mu_{x})\,\d x=-\frac14\int_{\R^2}\rho\log\rho-\frac12 D(\mu,\mu)\int_{\R^2}\rho+\frac14\int_{\R^2}\rho\log\lambda.$$
Hence we have proved that 
\begin{multline*}
\pscal{\sum_{1\leq i<j\leq N}w(x_i-x_j)}_\Psi\geq \frac12 D(\rho,\rho)-\frac14\int_{\R^2}\rho\log\rho-\frac12 D(\mu,\mu)\int_{\R^2}\rho\\+\frac14\int_{\R^2}\rho\log\lambda
+\int_{\R^2}\rho(x)\,D(\rho,\mu_{x}-\delta_x)\,\d x 
\end{multline*}
and it remains to bound the last term. 

Let $\nu_R=(\pi R)^{-2}\chi(|x|\leq R)$ be the (normalized) characteristic function of the disk of radius $R$. We have 
$$\big(\nu_R\ast w\big)(x)=w(x)-\frac{1}{2}\left(-\log\frac{|x|^2}{R^2}+\frac{|x|^2}{R^2}-1\right)\chi(|x|\leq R)$$
and thus
$$\big(\nu_R\ast w\big)(x)-w(x)\geq \log\left(\frac{|x|}{R}\right) \chi (|x|\leq R).$$
By scaling we find that
$$D(\rho,\mu_{x}-\delta_x)=D(\rho(\cdot +x),\nu_{R(x)}-\delta_0)\geq \int_{|y|\leq R(x)}\rho(y+x)\log\left(\frac{|y|}{R(x)}\right)\,\d y$$
and therefore
$$\int_{\R^2}\rho(x)\,D(\rho,\mu_{x}-\delta_x)\,\d x\geq\int_{\R^2}\d x\,\rho(x)\int_{|y|\leq R(x)}\d y\,\rho(y+x)\,\log\left(\frac{|y|}{R(x)}\right).$$

The following is similar to~\cite[Lemma 2] {Lieb-79}.
\begin{lemma}\label{lem:estim_M}
Let $f\in L^1(\R^2,\R^+)$ and let
\begin{equation}
M_f(0):=\sup_r\frac{1}{\pi r^2}\int_{|y|\leq r}f(y)\,\d y.
\end{equation}
Then we have for any $R>0$
\begin{equation}
-\int_{|y|\leq R}\d y\,f(y)\,\log\left(\frac{|y|}{R}\right)\leq \frac{\pi}{2} M_f(0)\, R^2.
\label{eq:estim_M} 
\end{equation}
\end{lemma}

\begin{proof}[Proof of Lemma~\ref{lem:estim_M}]
Let $\tilde f(r):=\int_{0}^{2\pi}f\big(r\cos(\theta),r\sin(\theta)\big)\,d\theta$ be the spherical average of $f$. Then we have
\begin{align*}
-\int_{|y|\leq R}\d y\,f(y)\,\log\left(\frac{|y|}{R}\right)&=-\int_{0}^Rr\tilde{f}(r)\,\log\left(\frac{r}{R}\right)\,\d r\\
&=\int_{0}^R\left(\int_0^r s\tilde{f}(s)\,ds\right)\,\frac{1}{r}\,\d r\\
&\leq \pi M_f \int_{0}^Rr\,\d r=\frac\pi2 M_f(0)\, R^2.
\end{align*}
In the previous estimate we have first integrated by parts, and then used 
$\int_0^r s\tilde{f}(s)\,ds=\int_{B_r}f\leq \pi r^2 M_f(0).$
\end{proof}

Using the previous estimate~\eqref{eq:estim_M} with $x\in \R^2$ fixed and recalling our choice $R(x)=\sqrt{\lambda(x)/\rho(x)}$, we get
$$\int_{|y|\leq R(x)}\d y\,\rho(y+x)\,\log\left(\frac{|y|}{R(x)}\right)\geq -\frac\pi2\,M_\rho(x)\,R(x)^2=-\frac{\pi\lambda(x)}{2\rho(x)}\,M_\rho(x)$$
where
$$M_\rho(x):=\sup_r\frac{1}{\pi r^2}\int_{|y|\leq r}\rho(y+x)\,\d y$$
is the Hardy-Littlewood maximal function of $\rho$. This gives the estimate
$$\int_{\R^2}\d x\,\rho(x)\int_{|y|\leq R(x)}\d y\,\rho(y+x)\,\log\left(\frac{|y|}{R(x)}\right)\geq -\frac\pi{2}\int_{\R^2}\lambda\, M_\rho$$
and our final bound is
\begin{multline*}
\pscal{\sum_{1\leq i<j\leq N}w(x_i-x_j)}_\Psi
\geq \frac12 D(\rho,\rho)-\frac14\int_{\R^2}\rho\log\rho-\frac12 D(\mu,\mu)\int_{\R^2}\rho\\+\frac14\int_{\R^2}\rho\log\lambda-\frac\pi2 \int_{\R^2}\lambda\,M_\rho.
\end{multline*}
Choosing $\lambda=\rho/M_\rho$, we have
\begin{align*}
\pscal{\sum_{1\leq i<j\leq N}w(x_i-x_j)}_\Psi
&\geq \frac12 D(\rho,\rho)-\frac14\int_{\R^2}\rho\log M_\rho-\Big( \frac12 D(\mu,\mu) + \frac{\pi}{2}\Big)\int_{\R^2}\rho.
\end{align*}
Now we estimate the error term
\begin{align*}
-\frac14\int_{\R^2}\rho\log M_\rho = -\frac14 \Big(\int_{\R^2}\rho \Big)\log \Big(\int_{\R^2}\rho \Big) - \frac14\Big(\int_{\R^2}\rho \Big) \int_{\R^2}f\log M_f
\end{align*}
with $f=\rho/ (\int_{\R^2} \rho)$. Using  $\epsilon f \log M_f\leq f(1+M_f^\epsilon) \le f + M_f^{1+\eps}$ we get 
$$ - \frac14\Big(\int_{\R^2}\rho \Big)\int_{\R^2}f\log M_f \ge - \frac{1}{4\eps} \int_{\R^2}\rho - \frac{1}{4}\Big(\int_{\R^2}\rho \Big) \int_{\R^2} M_f^{1+\eps}.$$
Using now~\cite{Stein-70}
$$\frac{1}{4}\int_{\R^2}M_f^{1+\epsilon}\leq \frac{C}\epsilon\int_{\R^2}f^{1+\epsilon}$$
we end up with the estimate
\begin{align*}
\pscal{\sum_{1\leq i<j\leq N}w(x_i-x_j)}_\Psi
&\geq \frac12 D(\rho,\rho)-\frac14 \Big(\int_{\R^2}\rho \Big)\log \Big(\int_{\R^2}\rho \Big) \\
&\quad - \Big( \frac12 D(\mu,\mu) + \frac{\pi}{2} +\frac{1}{4\eps}\Big)\int_{\R^2}\rho \\
&\quad- \frac{C}{\eps} \Big(\int_{\R^2}\rho \Big) \int_{\R^2} \left(\frac{\rho}{\int_{\R^2}\rho} \right)^{1+\epsilon} .
\end{align*}
\end{proof}



\end{document}